\documentclass[11pt,a4paper]{article}
\pdfoutput=1

\usepackage[margin=1in]{geometry}
\usepackage{amsmath,amssymb,amsthm,booktabs,color,enumitem,graphicx,url,xcolor}
\usepackage[numbers,sort&compress]{natbib}
\usepackage{todonotes}
\usepackage[algoruled,vlined,linesnumbered,nofillcomment]{algorithm2e}
\usepackage{microtype,hyperref}

\definecolor{citecolor}{HTML}{0000C0}
\definecolor{urlcolor}{HTML}{000080}

\hypersetup{
    colorlinks=true,
    linkcolor=black,
    citecolor=citecolor,
    filecolor=black,
    urlcolor=urlcolor,
}

\newtheorem{theorem}{Theorem}
\newtheorem{lemma}{Lemma}
\newtheorem{corollary}{Corollary}

\newtheorem{condition}{Condition}

\newcommand{\namedref}[2]{\hyperref[#2]{#1~\ref*{#2}}}
\newcommand{\sectionref}[1]{\namedref{Section}{#1}}

\newcommand{\theoremref}[1]{\namedref{Theorem}{#1}}
\newcommand{\corollaryref}[1]{\namedref{Corollary}{#1}}
\newcommand{\figureref}[1]{\namedref{Figure}{#1}}
\newcommand{\lemmaref}[1]{\namedref{Lemma}{#1}}

\newcommand{\algref}[1]{\namedref{Algorithm}{#1}}

\newcommand{\conditionref}[1]{\namedref{Condition}{#1}}

\newcommand{\N}{\mathbb N}
\newcommand{\R}{\mathbb R}
\newcommand{\BO}{\mathcal{O}}

\newenvironment{mycover}
               {\list{}{\listparindent 0pt
                        \itemindent    \listparindent
                        \leftmargin    0pt
                        \rightmargin   0pt
                        \parsep        0pt}%
                \raggedright
                \item\relax}
               {\endlist}

\begin{document}

\hypersetup{
    pdfauthor={Pankaj Khanchandani, Christoph Lenzen},
    pdftitle={Self-stabilizing Byzantine Clock Synchronization with Optimal
    Precision},
}

\begin{mycover}
{\LARGE \textbf{Self-stabilizing Byzantine Clock Synchronization\\
with Optimal Precision}\par}

\bigskip
\bigskip

\medskip
\textbf{Pankaj Khanchandani}

\smallskip
{\small Computer Engineering and Networks Laboratory (TIK), \\
ETH Zurich\par}

\bigskip

\textbf{Christoph Lenzen}

\smallskip
{\small Max Planck Institute for Informatics, \\
Saarland Informatics Campus,\par}

\end{mycover}

\bigskip
\paragraph{Abstract}
We revisit the approach to Byzantine fault-tolerant clock synchronization based
on approximate agreement introduced by Lynch and Welch. Our contribution is
threefold:
\begin{itemize}
  \item We provide a slightly refined variant of the algorithm yielding improved
  bounds on the skew that can be achieved and the sustainable frequency offsets.
  \item We show how to extend the technique to also synchronize clock rates.
  This permits less frequent communication without significant loss of
  precision, provided that clock rates change sufficiently slowly.
  \item We present a coupling scheme that allows to make these algorithms
  self-stabilizing while preserving their high precision. The scheme utilizes
  a low-precision, but self-stabilizing algorithm for the purpose of recovery.
\end{itemize}

\thispagestyle{empty}
\setcounter{page}{0}
\newpage

\section{Introduction}

When designing a synchronous distributed system, the most fundamental question
is how to generate and distribute the system clock. This task is mission
critical, both in terms of performance and reliability. With ever-growing
complexity of hardware, reliable high-performance clocking becomes increasingly
challenging; at the same time, the ramifications of clocking errors become
harder to predict.

Against this background, it might be unsurprising that fault-tolerant
distributed clock synchronization algorithms have found their way into
real-world systems with high reliability demands: the Time-Triggered Protocol
(TTP)~\cite{kopetz03} and FlexRay~\cite{flexray2005,Fuegger09} tolerate
\emph{Byzantine} (i.e., worst-case) faults and are utilized in cars and
airplanes. Both of these systems derive from the classic fault-tolerant
synchronization algorithm by Lynch and Welch~\cite{welch88}, which is based on
repeatedly performing approximate agreement~\cite{dolev86} on the time of the
next clock pulse. Another application domain with even more stringent
requirements is hardware for spacecraft and satellites. Here, a reliable system
clock is in demand despite frequent transient faults due to radiation. In
addition, quartz oscillators are prone to damage during launch, making the use
of less accurate, electronic oscillators preferable.

Unfortunately, existing implementations are not \emph{self-stabilizing,} i.e.,
do not guarantee automatic recovery from transient faults. This is essential for
the space domain, but also highly desirable in the systems utilizing TTP or
FlexRay. This claim is supported by the presence of various mechanisms that
monitor the nodes and perform resets in case of observed faulty behavior in both
protocols. Thus, it is of interest to devise synchronization algorithms that
stabilize on their own, instead of relying on monitoring techniques: these need
to be highly reliable as well, or their failure may bring down the system due to
erroneous detection of or response to faults.

Against this backdrop, in this work we set out to answer the following
questions:
\begin{enumerate}
  \item Can the guarantees of \cite{welch88} be further improved? In particular,
  how does the approach perform if the (relative) phase drift of the local clock
  sources are larger than for typical quartz oscillators?
  \item Under which circumstances is it useful to apply the technique also to
  frequencies, i.e., algorithmically adjust clock rates?
  \item Can the solution be made self-stabilizing?
\end{enumerate}

\paragraph{Our Contribution.}
We obtain promising answers to the above questions, in the sense that
conceptually simple (i.e., implementation-friendly!) variations on the
Lynch-Welch approach achieve excellent performance guarantees. Specifically, we
obtain the following main results.
\begin{enumerate}
  \item We present a refined analysis of a variant of the Lynch-Welch algorithm.
  We show that the algorithm converges to a steady-state error
  $E\in \BO((\vartheta-1)T+U)$\,, where hardware clock rates are between $1$ and
  $\vartheta$, messages take between $d-U$ and $d$ time to arrive at their
  destination, and $T\in \Omega(d)$ is the (nominal) time between consecutive
  clock pulses (i.e., the time required for a single approximate agreement
  step). This works even for very poor local clock sources: it suffices if
  $\vartheta\leq 1.1$, although the skew bound goes to infinity as $\vartheta$
  approaches this critical value; for $\vartheta\leq 1.01$, the bound is fairly
  close to $2(\vartheta-1)T+4U$.\footnote{For comparison, the critical value in
  \cite{welch88} is smaller than $1.025$, i.e., we can handle a factor $4$
  weaker bound on $\vartheta-1$. Non-quartz oscillators used in space
  applications, where temperatures vary widely, may have $\vartheta$ close to
  this value, cf.~\cite{oscis}.}
  \item We give a second algorithm that interleaves approximate agreement on
  clock \emph{rates} with the phase (i.e., clock offset) correction scheme. If
  the clocks are sufficiently stable, i.e., the maximum rate of change $\nu$ of
  clock rates is sufficiently small, this enables to significantly extend $T$
  (and thus decrease the frequency of communication) without substantially
  affecting skews. Provided that $\vartheta$ is not too large, for any $T$
  satisfying $\max\{(\vartheta-1)^2 T,\nu T^2\}\ll U$, it is possible to
  guarantee a skew of $\BO(U)$.
  \item We introduce a generic approach that enables to couple either of these
  algorithms to FATAL~\cite{dolev14fatal,dolev14}. FATAL is a self-stabilizing
  synchronization algorithm, but in comparison suffers from poor performance.
  The coupling scheme permits to combine the best of both worlds, namely the
  self-stabilization properties of FATAL with the small skew of the Lynch-Welch
  synchronization scheme.
\end{enumerate}
On the technical side, the first two results require little innovation compared
to prior work. However, it proved challenging to obtain clean, easy-to-implement
algorithms that are amenable to a tractable analysis and achieve tight skew
bounds. This is worthwhile for two reasons: (1) there is strong indication that
the approach has considerable practical merit,\footnote{A prototype FPGA
implementation achieves $182\,$ps skew~\cite{HKL16}, which is suitable for
generating a system clock.} and (2) no readily usable mathematical analysis of
the frequency correction scheme exists in the literature.\footnote{The framework
in \cite{schossmaier98,schossmaier99} addresses frequency correction, but
substantial specialization of the framework, including its mathematical
analysis, would be required to achieve good constants in the bounds.} In fact,
the second algorithm we present differs from FlexRay (which also aims to adjust
frequencies) in a crucial point. In order to avoid that the approximate
agreement scheme is rendered ineffective because nodes reach the imposed limits
on adjusting their frequency,\footnote{Constraining feasible clock rates is
necessary to avoid that measurement errors result in clocks speeding up or
slowing down arbitrarily over time.} we add a correction slowly pulling back
nodes' frequencies to the nominal rate. Without this provision, it is
straightforward to construct executions in which, e.g., the majority of the
nodes runs too fast for another node to sufficiently adjust its clock rate to
match their speed. This means that, in the worst case, FlexRay's frequency
correction is futile.

In contrast, the coupling scheme we use to combine our non-stabilizing
algorithms with FATAL showcases a novel technique of independent interest. We
leverage FATAL's clock ``beats'' to effectively (re-)initialize the
synchronization algorithm we couple it to. Here, care has to be taken to avoid
such resets from occurring during regular operation of the Lynch-Welch scheme,
as this could result in large skews or even spurious clock pulses. The solution is
a feedback mechanism that enables the synchronization algorithm to actively
trigger the next beat of FATAL at the appropriate time. FATAL stabilizes
regardless of how these feedback signals behave, while actively triggering beats
ensures that all nodes pass the checks which, if failed, trigger the respective
node being reset.

While a specific interface is required from the stabilizing algorithm to permit
this approach, it seems likely that most, if not all, self-stabilizing
synchronization algorithms could be modified to provide it. Thus, we consider
the technique a highly useful separation of the tasks to achieve small skews and
to ensure (fast) stabilization.

\paragraph{Organization of the paper.}
After presenting related work and the model, we proceed in the order of the main
results listed above: phase synchronization (\sectionref{sec:phase}), frequency
synchronization (\sectionref{sec:frequency}), and finally the coupling scheme
adding self-stabilization (\sectionref{sec:self}). \sectionref{sec:conclusions}
concludes the paper.

\section{Related Work}\label{sec:related}

TTP~\cite{kopetz03} and FlexRay~\cite{flexray2005,Fuegger09} are both
implemented in software (barring minor hardware components). This is sufficient
for their application domains: the goal here is to enable synchronous
communication between hardware components at frequencies in the megahertz range.
Solutions fully implemented in hardware are of interest for two reasons. First,
having to implement the full software abstraction dramatically increases the
number of potential reasons for a node to fail -- at least from the point of
view of the synchronization algorithm. A slim hardware implementation is thus
likely to result in a substantially higher degree of reliability of the clocking
mechanism. Second, if higher precision of synchronization is required, the
significantly smaller delays incurred by dedicated hardware make it possible to
meet these demands.

Apart from these issues, the complexity of a software solution renders TTP and
FlexRay unsuitable as fault-tolerant clocking schemes for VLSI circuits. The
DARTS project~\cite{darts,FS12} aimed at developing such a scheme, with the goal
of coming up with a robust clocking method for space applications. Instead of
being based on the Lynch-Welch approach, it implements the fault-tolerant
synchronization algorithm by Srikanth and Toueg~\cite{ST87}. Unfortunately,
DARTS falls short of its design goals in two ways. On the one hand, the
Srikanth-Toueg primitive achieves skews of $\Theta(d)$, which tend to be
significantly larger than those attainable with the Lynch-Welch
approach.\footnote{The maximum delay $d$ tends to be at least one or two orders
of magnitude larger than the delay uncertainty~$U$.} Accordingly, the
operational frequency DARTS can sustain (without large communication buffers
and communication delays of multiple logical rounds) is in the range of
$100\,$MHz, i.e., about an order of magnitude smaller than typical system
speeds. Moreover, DARTS is not self-stabilizing. This means that DARTS
-- just like TTP and FlexRay -- is unlikely to successfully cope with high rates
of transient faults. Worse, the rate of transient faults will scale with the
number of nodes (and thus sustainable faults). For space environments, this
implies that adding fault-tolerance without self-stabilization cannot be
expected to increase the reliability of the system at all.

These concerns inspired follow-up work seeking to overcome these downsides of
DARTS. From an abstract point of view, FATAL~\cite{dolev14fatal,dolev14} can be
interpreted as another incarnation of the Srikanth-Toueg approach. However,
FATAL combines tolerance to Byzantine faults with self-stabilization in $\BO(n)$
time with probability $1-2^{-\Omega(n)}$; after recovery is complete, the
algorithm maintains correct operation deterministically. Like DARTS, FATAL and
the substantial line of prior work on Byzantine self-stabilizing synchronization
algorithms (e.g., \cite{DD06,DolWelSSBYZCS04}) cannot achieve better clock skews
than $\Theta(d)$. The key motivation for the present paper is to combine the
better precision achieved by the Lynch-Welch approach with the
self-stabilization properties of FATAL.

Concerning frequency correction, little related work exists. A notable exception
is the extension of the interval-based synchronization framework to rate
synchronization~\cite{schossmaier98,schossmaier99}. In principle, it seems
feasible to derive similar results by specialization and minor adaptions of this
powerful machinery to our setting. Unfortunately, apart from the technical
hurdles involved, an educated guess (based on the amount of necessary
specialization and estimates that need to be strengthened) result in worse
constants and more involved algorithms, and it is unclear whether our approach
to self-stabilization can be fitted to this framework. However, it is worth
noting that the overall proof strategies for the (non-stabilizing) phase and
frequency correction algorithms bear notable similarities to this generic
framework: separately deriving bounds on the precision of measurements, plugging
these into a generic convergence argument, and separating the analysis of
frequency and phase corrections.

Coming to lower bounds and impossibility results, the following is known.
\begin{itemize}
  \item In a system of $n$ nodes, no algorithm can tolerate $\lceil n/3\rceil$
  Byzantine faults. All mentioned algorithms are optimal in that they tolerate
  $\lceil n/3\rceil -1$ Byzantine faults~\cite{dolev1986possibility}.
  \item To tolerate this number of faults, $\Omega(n^2)$ communication links are
  required.\footnote{If a node has fewer than $2f+1$ neighbors in a system
  tolerating $f$ faults, it cannot distinguish whether it synchronizes to a
  group of $f$ correct or $f$ faulty neighbors.} All mentioned algorithms assume
  full connectivity and communicate by broadcasts (faulty nodes may not adhere
  to this). Less well-connected topologies are outside the scope of this work.
  \item The worst-case precision of an algorithm cannot be better than
  $(1-1/n)U$ in a network where communication delays may vary by
  $U$~\cite{lundelius84}. In the fault-free case and with $\vartheta-1$
  sufficiently small, this bound can be almost matched
  (cf.~\sectionref{sec:phase}); all variants of the Lynch-Welch approach match
  this bound asymptotically granted sufficiently accurate local clocks.
  \item Trivially, the worst case precision of any algorithm is at least
  $(\vartheta-1)T$ if nodes exchange messages every $T$ time units. In the
  fault-free case, this is essentially matched by our phase correction algorithm
  as well.
  \item With faults, the upper bound on the skew of the algorithm increases by
  factor $1/(1-\alpha)$, where $\alpha \approx 1/2$ if $\vartheta \approx 1$. It
  appears plausible that this is optimal under the constraint that the
  algorithm's resilience to Byzantine faults is optimal, due to a lower bound
  on the convergence rate of approximate agreement~\cite{dolev86}.
\end{itemize}
Overall, the resilience of the presented solution to faults is optimal, its
precision asymptotically optimal, and it seems reasonable to assume that there
is little room for improvement in this regard. In contrast, no non-trivial lower
bounds on the stabilization time of self-stabilizing fault-tolerant
synchronization algorithms are known. It remains an open question whether it is
possible to achieve stabilization within $o(n)$ time.

\section{Model}\label{sec:model}

We assume a fully connected system of $n$ nodes, up to $f:=\lfloor
(n-1)/3\rfloor$ of which may be Byzantine faulty (i.e., arbitrarily deviate
from the protocol). We denote by $V$ the set of all nodes and by $C\subseteq V$
the subset of \emph{correct} nodes, i.e., those that are not faulty.

Communication is by broadcast of ``pulses,'' which are messages without content:
the only information conveyed is when a node transmitted a pulse. Nodes can
distinguish between senders; this is used to distinguish the case of multiple
pulses being sent by a single (faulty) node from multiple nodes sending one
pulse each. Note that faulty nodes are not bound by the broadcast restriction,
i.e., may send a pulse to a subset of the nodes only. The system is
semi-synchronous. A pulse sent by node $v\in C$ at (Newtonian) time
$p_v\in \R_0^+$ is received by node $w\in C$ at time $t_{vw}\in
[p_v+d-U,p_v+d]$; we refer to $d$ as the \emph{maximum message delay} (or,
chiefly, delay) and to $U$ as the \emph{delay uncertainty} (or, chiefly,
uncertainty).

For these timing guarantees to be useful to an algorithm, the nodes must have a
means to measure the progress of time. Each node $v\in C$ is equipped with a
hardware clock $H_v$, which is modeled as a strictly increasing function
$H_v:\R^+_0\to \R^+_0$. We require that there is a constant $\vartheta>1$ such
that for all times $t<t'$, it holds that
\begin{equation*}
t'-t\leq H_v(t')-H_v(t)\leq \vartheta (t'-t)\,,
\end{equation*}
i.e., the hardware clocks have bounded drift.\footnote{It is common to define
the drift symmetrically, i.e., $(1-\rho)(t'-t)\leq H_v(t')-H_v(t)\leq
(1+\rho)(t'-t)$ for some $0<\rho< 1$. For $\rho\ll 1$ and
$\vartheta\approx 1$, up to minor order terms this is equivalent to setting
$\rho:=(\vartheta-1)/2$ and rescaling the real time axis by factor $1-\rho$. The
one-sided formulation results in less cluttered notation.} We remark that our
results can be easily translated to the case of discrete and bounded
clocks.\footnote{Discretization can be handled by re-interpreting the
discretization error as part of the delay uncertainty. All our algorithms use
the hardware clock exclusively to measure bounded time differences.} We refer
to $H_v(t)$ as the \emph{local time} of $v$ at time $t$.

Executions are event-based, where an event at node $v$ is the reception of a
message, a previously computed (and stored) local time being reached, or the
initialization of the algorithm. A node may then perform computations and
possibly send a pulse. For simplicity, we assume that these operations take zero
time; adapting our results to account for computation time is straightforward.

\paragraph{Problem.}
A clock synchronization algorithm generates distinguished events or
\textit{clock pulses} at times $p_v(r)$ for $r \in \N$ and $v \in
C$ so that the following conditions are satisfied for all $r\in \N$.
\begin{enumerate}
  \item $\forall v,w\in C:~|p_v(r) - p_w(r)| \leq e(r)$
  \item $\forall v\in C:~A_{\min} \leq p_v(r + 1) - p_v(r) \leq A_{\max}$
\end{enumerate}
The first requirement is a bound on the synchronization error between the
$r^{th}$ clock ticks; naturally, it is desired that $e(r)$ is as small as
possible. The second requirement is a bound on the time between consecutive
clock ticks, which can be translated to a bound on the frequency of the clocks;
here, the goal is that $A_{\min}/A_{\max}\approx 1$. The \emph{precision} of the
algorithm is measured by the steady state error\footnote{Typically, $e(r)$ is a
monotone sequence, implying that simply $E=\lim_{r\to \infty}e(r)$.}
\begin{equation*}
E := \lim_{r'\to\infty}\sup_{r\geq r'}\{e(r)\}\,.
\end{equation*}
Self-stabilization will be introduced and discussed in \sectionref{sec:self}.

\section{Phase Synchronization Algorithm}\label{sec:phase}

Our basic algorithm is a variant of the one by Lynch and Welch~\cite{welch88},
which synchronizes clocks by simulating perpetual synchronous approximate
agreement~\cite{dolev86} on the times when clock pulses should be generated. We
diverge only in terms of communication: instead of round numbers, nodes
broadcast content-free pulses. Due to sufficient waiting times between pulses,
during regular operation received messages from correct nodes can be correctly
attributed to the respective round. In fact, the primary purpose of transmitting
round numbers in the Lynch-Welch algorithm is to add recovery properties. Our
technique for adding self-stabilization (presented in \sectionref{sec:self})
leverages the pulse synchronization algorithm from~\cite{dolev14,dolev14fatal}
instead, which requires to broadcast constant-sized messages only.

Before presenting the algorithm and its analysis in
Sections~\ref{sec:basic_algo} and~\ref{sec:basic_analysis}, respectively, we
revisit some basic properties of the technique for approximate agreement
introduced in~\cite{dolev86} in the context used here. The results in this
section are derivatives of the ones from~\cite{dolev86,welch88}, but adapting
them to our setting and notation is essential for deriving our main results in
Sections~\ref{sec:frequency} and~\ref{sec:self}.

\subsection{Properties of Approximate Agreement Steps}\label{sec:approx}

Abstractly speaking, the synchronization performs approximate agreement steps in
each (simulated synchronous) round. In approximate agreement, each node is given
an input value and the goal is to let nodes determine values that are close to
each other and within the interval spanned by the correct nodes' inputs.

In the clock synchronization setting, there is the additional obstacle that the
communicated values are points in time. Due to delay uncertainty and drifting
clocks, the communicated values are subject to a (worst-case) perturbation of at
most some $\delta\in \R^+_0$. We will determine $\delta$ later in our analysis
of the clock synchronization algorithms; we assume it to be given for now. The
effect of these disturbances is straightforward: they may shift outputs by at
most $\delta$ in each direction, increasing the range of the outputs by an
additive $2\delta$ in each step (in the worst case).

\algref{alg:approxAgree} describes an approximate agreement step from the point
of view of node $v\in C$. When implementing this later on, we need to make use
of timing constraints to ensure that (i) correct nodes receive each other's
messages in time to perform the associated computations and (ii) correct nodes'
messages can be correctly attributed to the round to which they belong.
\figureref{fig:algorithm_approx} depicts how a round unfolds assuming that
these timing constraints are satisfied.

\begin{algorithm}\label{alg:approxAgree}
	// node $v$ is given input value $x_v$\;
	broadcast $x_{v}$ to all nodes (including self)\;
	// if $w\in C$, the received value $\hat{x}_{wv}\in [x_w-\delta,x_w+\delta]$\;
	receive first value $\hat{x}_{wv}$ from each node $w$ ($\hat{x}_{wv}:=x_v$ if
	no message from $w$ received)\;
	$S_{v}\gets \{\hat{x}_{wv}\,|\,w\in V\}$\;
	denote by $S^k_v$ the $k^{th}$ element of $S_{v}$ w.r.t.\ ascending order\;
	$y_{v} \gets \dfrac{S_{v}^{f+1} + S_{v}^{n - f}}{2} $\;
	\Return $y_v$\;
\caption{Approximate agreement step at node $v \in \textbf{C}$ (with
synchronous message exchange).}
\end{algorithm}

\begin{figure}[t!]
	\centering
	\def\svgwidth{\columnwidth}
	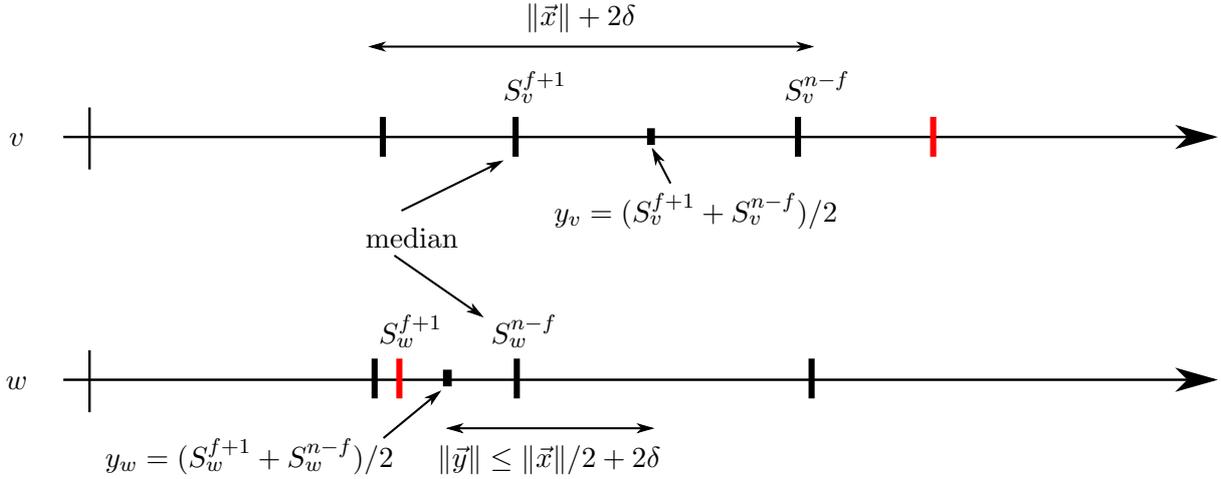
	\caption{An execution of \algref{alg:approxAgree} at nodes $v$ and $w$ of a
	system consisting of $n=4$ nodes. There is a single faulty node and its values
	are indicated in red. Note that the ranges spanned by the values received from
	non-faulty nodes are \emph{almost} identical; the difference originates in the
	perturbations of up to $\delta$.}
   \label{fig:algorithm_approx}
\end{figure}

Denote by $\vec{x}$ the $|C|$-dimensional vector of correct nodes' inputs,
i.e., $(\vec{x})_v=x_v$ for $v\in C$. The \textit{diameter} $\|\vec{x}\|$ of
$\vec{x}$ is the difference between the maximum and minimum components of
$\vec{x}$. Formally,
\begin{gather*}
\|\vec{x}\| := \max_{v\in C}\{x_v\} - \min_{v\in C}\{x_v\}.
\end{gather*} 
We will use the same notation for other values, e.g.\ $\vec{y}$ and
$\|\vec{y}\|$. For simplicity, we assume that $|C|=n-f$ in the following; all
statements can be adapted by replacing $n-f$ with $|C|$ where appropriate.

Consider the special case of $\delta=0$. Intuitively, \algref{alg:approxAgree}
discards the smallest and largest $f$ values each to ensure that values from
faulty nodes cannot cause outputs to lie outside the range spanned by the
correct nodes' values. Afterwards, $y_v$ is determined as the midpoint of the
interval spanned by the remaining values. Since $f<n/3$, i.e., $n-f\geq 2f+1$,
the median of correct nodes' values is part of all intervals computed by correct
nodes. From this, it is easy to see that $\|\vec{y}\|\leq \|\vec{x}\|/2$, see
\figureref{fig:algorithm_approx}. For $\delta>0$, we simply observe that the
resulting values $y_v$, $v\in C$, are shifted by at most $\delta$ compared to the case
where $\delta=0$, resulting in $\|\vec{y}\|\leq \|\vec{x}\|/2+2\delta$. We now
prove these properties.

\begin{lemma}\label{lem:validity}
\begin{equation*}
\forall v\in C:\,\min_{w\in C}\{x_w\}-\delta\leq y_v \leq \max_{w\in
C}\{x_w\}+\delta\,.
\end{equation*}
\end{lemma}
\begin{proof}
As there are at most $f$ faulty nodes, for $v\in C$ we have that
\begin{equation*}
S_v^{f+1}\geq \min_{w\in C}\{\hat{x}_{wv}\}\geq \min_{w\in C}\{x_w\}-\delta\,.
\end{equation*}
Analogously, $S_v^{n-f}\leq \max_{w\in C}\{x_w\}+\delta$. We conclude that
\begin{equation*}
\min_{w\in C}\{x_w\}-\delta\leq S_v^{f+1}\leq
\frac{S_v^{f+1}+S_v^{n-f}}{2}=y_v\leq S_v^{n-f} \leq \max_{w\in C}\{x_w\}+\delta\,.\qedhere
\end{equation*}
\end{proof}

\begin{corollary}\label{cor:max_correction}
	$\max_{v\in C}\{|y_v-x_v|\} \leq \|\vec{x}\|+\delta$.
\end{corollary}

\begin{lemma}\label{lem:convergence}
$\|\vec{y}\|\leq \|\vec{x}\|/2+2\delta$.
\end{lemma}
\begin{proof}
We show the claim for $\delta=0$ first, i.e., $\hat{x}_{wv}=x_w$ for all
$v,w\in C$. Denote by $x^k$ the $k^{th}$ element of $\vec{x}$ w.r.t.\ ascending
order. Since $f<n/3$, we have that $n-f\geq 2f+1$. Hence, for all $v\in C$,
\begin{equation*}
x^1\leq S_v^{f+1}\leq x^{f+1}\leq S_v^{2f+1}\leq S_v^{n-f}\leq x^{n-f}\,.
\end{equation*}
For any $v,w\in C$, it follows that
\begin{equation*}
y_v-y_w=\frac{S_{v}^{f+1} - S_{w}^{f+1} + S_{v}^{n-f}-S_{w}^{n-f}}{2}
\leq
\frac{x^{f+1}-x^1+x^{n-f}-x^{f+1}}{2}=\frac{x^{n-f}-x^1}{2}=\frac{\|\vec{x}\|}{2}\,.
\end{equation*}
Symmetrically, we have that $y_w-y_v\leq \|\vec{x}\|/2$ and thus $|y_v-y_w|\leq
\|\vec{x}\|/2$. As $v,w\in C$ were arbitrary, this yields $\|\vec{y}\|\leq
\|\vec{x}\|/2$ (under the assumption that $\delta=0$).

For the general case, observe that $S_{v}^{f+1}$, $S_{w}^{f+1}$, $S_{v}^{n-f}$,
and $S_{w}^{n-f}$ each can be changed by at most $\delta$. This can affect
$(S_{v}^{f+1} - S_{w}^{f+1} + S_{v}^{n-f}-S_{w}^{n-f})/2$ by at most
$4\delta/2=2\delta$; the claim follows.
\end{proof}

\subsection{Algorithm}\label{sec:basic_algo}

\algref{alg:basic} shows the pseudocode of the phase synchronization algorithm
at node $v\in C$. It implements iterative approximate agreement steps on the
times when to send pulses. The algorithm assumes that the nodes are initialized
within a (local) time window of size $F$. In each round $r\in \N$, the nodes
estimate the phase offset of their pulses\footnote{Note that we divide the
measured local time differences by factor $(\vartheta+1)/2$, the average of the
minimum and maximum clock rates. This is an artifact of our more
notation-friendly ``one-sided'' definition of hardware clock rates from
$[1,\vartheta]$; in an implementation, one simply reads the hardware clocks
(which exhibit symmetric error) without any scaling.} and then compute an
according phase correction $\Delta_v(r)$. \figureref{fig:algorithm_basic}
illustrates how a round of the algorithm plays out.

\begin{algorithm}[t!]\label{alg:basic}
\caption{Phase synchronization algorithm, code for node~$v\in C$. Time $t_v(r)$,
$r\in \N_0$, is the time when round $r+1$ starts.}
// $H_w(0)\in [0,F)$ for all $w\in V$\\
wait until time $t_v(0)$ with $H_v(t_v(0))=F$\;
\ForEach{round $r \in \N$} {
	start listening for messages\;
	wait until local time $H_v(t_v(r-1))+\tau_1(r)$;\hfill// all nodes are in round
	$r$\label{line:wait1}\\
	broadcast clock pulse to all nodes (including self)\; 
	wait until local time $H_v(t_v(r-1))+\tau_1(r)+\tau_2(r)$;\hfill
	// correct nodes' messages arrived\label{line:wait2}\\
	\For{each node $w\in V$}{
		$\tau_{wv}:= H_v(t_{wv})$, where first message from $w$ received at $t_{wv}$
		($\tau_{wv}:=\infty$ if none received)\;
	}
	$S_{v} \gets \{2(\tau_{wv}-\tau_{vv})/(\vartheta+1)\mid w\in V\}$ (as
	multiset)\; let $S_v^k$ denote the $k^{th}$ smallest element of $S_v$\; 
	$\Delta_{v}(r) \gets \dfrac{ S_{v}^{f + 1} + S_{v}^{n - f}}{2}$\;
	// $T(r)$ denotes the nominal length of round $r$\\
	wait until time $t_v(r)$ with $H_{v}(t_{v}(r))=H_{v}(t_{v}(r-1)) + T(r) -
	\Delta_{v}(r)$\;
}
\end{algorithm}

\begin{figure}[t!]
	\centering
	\def\svgwidth{\columnwidth}
	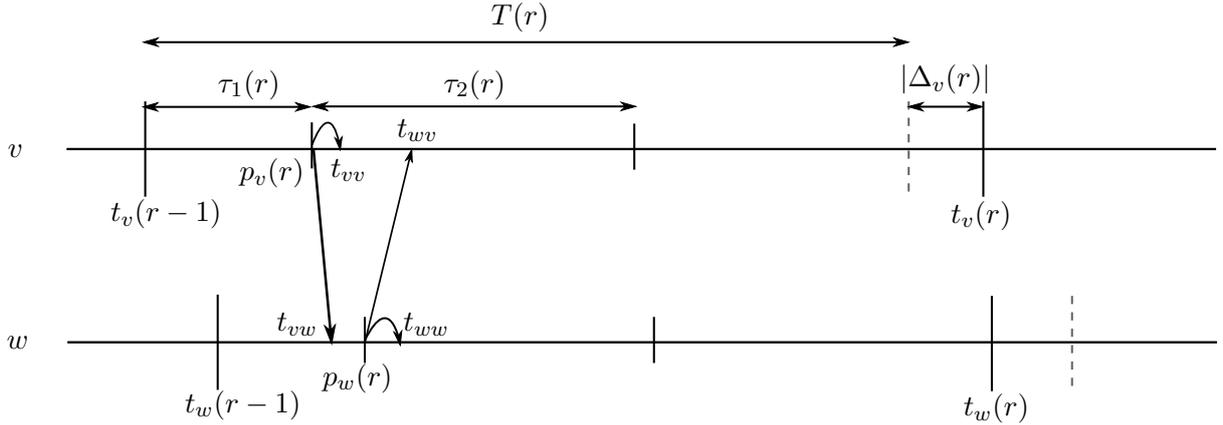
	\caption{A round of \algref{alg:basic} from the point of view of nodes $v$ and
	$w$. Note that the durations marked on the horizontal axis are measured using
	the local hardware clock.}
   \label{fig:algorithm_basic}
\end{figure}

To fully specify the algorithm, we need to determine how long the waiting
periods in each round are (in terms of local time), which will be given as
$\tau_1(r)$, $\tau_2(r)$, and $T(r)-\Delta(r)-\tau_1(r)-\tau_2(r)$. Here, we
must ensure for all $r\in \N$ that
\begin{enumerate}
  \item for all $v,w\in C$, the message that $v$ broadcasts at time
  $t_v(r-1)+\tau_1(r)$ is received by $w$ at a local time from
  $[H_w(t_w(r-1)),H_w(t_w(r-1))+\tau_1(r)+\tau_2(r)]$ and
  \item for all $v\in C$, $T(r)-\Delta_v(r)\geq \tau_1(r)+\tau_2(r)$, i.e., $v$
  computes $H_v(t_v(r))$ \emph{before} time $t_v(r)$.
\end{enumerate}
If these conditions are satisfied at all correct nodes, we say that \emph{round
$r$ is executed correctly}, and we can interpret the round as an approximate
agreement step in the sense of \sectionref{sec:approx}. We will show in the next
section that the following condition is sufficient for all rounds to be executed
correctly.
\begin{condition}\label{cond:constraints}
Define $e(1):=F+(1-1/\vartheta)\tau_1(1)$ and inductively for all $r\in \N$ that
\begin{equation*}
e(r+1):=\frac{2\vartheta^2+5\vartheta-5}{2(\vartheta+1)}\,e(r)
+(3\vartheta-1)U+\left(1-\frac{1}{\vartheta}\right)(T(r)+\tau_1(r+1)-\tau_1(r))\,.
\end{equation*}
We require for all $r\in \N$ that
\begin{align*}
\tau_1(r)&\geq \vartheta e(r)\\
\tau_2(r)&\geq \vartheta(e(r)+d)\\
T(r)&\geq \tau_1(r)+\tau_2(r)+\vartheta(e(r)+U)\,.
\end{align*}
\end{condition}
Here, $e(r)$ is a bound on the synchronization error in round $r$, i.e., we will
show that $\|\vec{p}(r)\|\leq e(r)$ for all $r\in \N$, provided
\conditionref{cond:constraints} is satisfied. \conditionref{cond:constraints}
cannot be satisfied for arbitrary $\vartheta>1$ such that $e(r)$ is bounded
independently of $r$. The intuition is that rounds must be long enough to ensure
that all pulses from correct nodes are received (i.e., at least $\vartheta
e(r)$), but during this time additional error is built up by drifting clocks; if
the approximate agreement step cannot overcome this \emph{relative} skew
increase, round $r+1$ has to be even longer, and so on. However, any
$\vartheta\leq 1.1$ can be sustained.
\begin{lemma}\label{lem:solvable}
\conditionref{cond:constraints} can be satisfied such that $\lim_{r\to \infty}
e(r)<\infty$ if
\begin{equation*}
\alpha:=\frac{6\vartheta^2+5\vartheta-9}{2(\vartheta+1)(2-\vartheta)}<1\,.
\end{equation*}
In this case, we can achieve
\begin{equation*}
\lim_{r\to \infty} e(r) \leq \frac{(\vartheta-1)d+(4\vartheta
-2)U}{(2-\vartheta)(1-\alpha)}\,.
\end{equation*}
\end{lemma}
\begin{proof}
By plugging $e(1)$ into the inequality for $\tau_1(1)$, we see that we may choose
$\tau_1(1)<\infty$ if and only if $\vartheta<2$. Assuming that this is the
case, we choose to satisfy all inequalities with equality, yielding for $r\in
\N$ that
\begin{align*}
\tau_1(r)&=\vartheta e(r)\\
T(r)&= \vartheta(3e(r)+d+U)\\
e(r+1)&= \frac{6\vartheta^2+5\vartheta-9}{2(\vartheta+1)(2-\vartheta)}\,e(r)
+\frac{(\vartheta-1) d}{2-\vartheta} + \frac{(4\vartheta-2)U}{2-\vartheta}
=\alpha e(r)+\frac{(\vartheta-1) d}{2-\vartheta} + \frac{(4\vartheta-2)U}{2-\vartheta}\,.
\end{align*}
Thus,
\begin{equation*}
\lim_{r\to \infty} e(r)= \lim_{r\to
\infty}\left(\alpha^{r-1}e(1)+\sum_{r'=0}^{r-1}\alpha^{r'}
\left(\frac{(\vartheta-1)d+(4\vartheta-2)U}{2-\vartheta}\right)\right)
=\frac{(\vartheta-1)d+(4\vartheta-2)U}{(2-\vartheta)(1-\alpha)}\,,
\end{equation*}
where the second equality holds because $\alpha<1$. Because $\alpha<1$ is a
stricter constraint on $\vartheta$ than $\vartheta<2$, this completes the proof.
\end{proof}
Several remarks are in order.
\begin{itemize}
  \item $\alpha$ goes to $1/2$ as $\vartheta$ goes to $1$. For $\vartheta=1.01$,
  we already have that $\alpha\approx 0.55$. Thus, the approach can support
  fairly large phase drifts.
  \item For $\vartheta\approx 1$, we have that $\lim_{r\to \infty} e(r)\approx
  4U+2(\vartheta-1)d$. From \corollaryref{cor:est}, one can see that if
  $(\vartheta-1)d\ll U$, this can be reduced to $\lim_{r\to \infty} e(r)\approx
  2U$.
  \item The lower bound by Lynch and Welch~\cite{lundelius84} shows that this is
  optimal up to factor $2$. It is straightforward to verify that in the
  fault-free case with $\vartheta=1$, the algorithm attains the lower bound.
  \item The convergence is exponential, i.e., for any $\varepsilon>0$ we have
  that $e(r)\leq (1+\varepsilon)\lim_{r\to \infty} e(r)$ for all $r\geq
  r_{\varepsilon}\in \Theta(\log F/(\varepsilon \lim_{r\to \infty} e(r)))$.
\end{itemize}

\subsection{Analysis}\label{sec:basic_analysis}

In this section, we prove that \conditionref{cond:constraints} is indeed
sufficient to ensure that $\|\vec{p}(r)\|\leq e(r)$ for all $r\in \N$.
In the following, denote by $\vec{p}(r)$, $r\in \N_0$, the vector of times when
nodes $v\in C$ broadcast their $r^{th}$ pulse, i.e.,
$H_v(p_v(r))=H_v(t_v(r-1))+\tau_1(r)$. If $v\in C$ takes note of the pulse from
$w\in C$ in round $r$, the corresponding value $\tau_{wv}-\tau_{vv}$ can be
interpreted as inexact measurement of $p_w(r)-p_v(r)$. This is captured by the
following lemma, which provides precise bounds on the incurred error.
\begin{lemma}\label{lem:est}
Suppose $v\in C$ receives the pulses from both $w\in C$ and itself in round $r$
at a time from $[H_v(t_v(r-1)),H_v(t_v(r-1))+\tau_1(r)+\tau_2(r)]$. Then
\begin{equation*}
\left|\frac{2(\tau_{wv}-\tau_{vv})}{\vartheta+1}-(p_w(r)-p_v(r))\right|<
\vartheta U + \frac{\vartheta-1}{\vartheta+1}\|\vec{p}(r)\|\,,
\end{equation*}
where $\tau_{wv}$ and $\tau_{vv}$ denote the values of the respective variables
in the algorithm in round $r$.
\end{lemma}
\begin{proof}
Denote by $t_{uv}$ the time when $v$ receives the pulse from $u\in \{v,w\}$. The
communication model guarantees that $t_{uv}\in [p_u(r)+d-U,p_u(r)+d]$. Thus,
\begin{equation}\label{eq:uncertainty}
\tau_{uv}=H_v(t_{uv})\in [H_v(p_u(r)+d-U),H_v(p_u(r)+d)]\subseteq
H_v(p_u(r)+d-U/2)\pm \frac{\vartheta U}{2}\,.
\end{equation}
Moreover, if $p_w(r)-p_v(r)\geq 0$, the bounds on the hardware clock speed
guarantee that
\begin{equation*}
\frac{2(p_w(r)-p_v(r))}{\vartheta+1}\leq
\frac{2(H_v(p_w(r)+d-U/2)-H_v(p_v(r)+d-U/2))}{\vartheta+1}\leq
\frac{2\vartheta(p_w(r)-p_v(r))}{\vartheta+1}
\end{equation*}
and thus
\begin{align*}
\frac{(1-\vartheta)(p_w(r)-p_v(r))}{\vartheta+1}
&\leq
\frac{2(H_v(p_w(r)+d-U/2)-H_v(p_v(r)+d-U/2))}{\vartheta+1}-(p_w(r)-p_v(r))\\
&\leq \frac{(\vartheta-1)(p_w(r)-p_v(r))}{\vartheta+1}\,.
\end{align*}
Since $|p_w(r)-p_v(r)|\leq \|\vec{p}(r)\|$ by definition, this yields that
\begin{equation}\label{eq:drift}
\left|\frac{2(H_v(p_w(r)+d-U/2)-H_v(p_v(r)+d-U/2))}{\vartheta+1}-(p_w(r)-p_v(r))\right|
\leq \frac{\vartheta-1}{\vartheta+1}\|\vec{p}(r)\|\,.
\end{equation}
This bound also holds in case $p_w(r)-p_v(r)<0$, as we can switch the roles of
$v$ and $w$ in the above inequalities. We conclude that
\begin{align*}
&\left|\frac{2(\tau_{wv}-\tau_{vv})}{\vartheta+1}-(p_w(r)-p_v(r))\right|\\
&\qquad\leq \frac{2}{\vartheta+1}
(|\tau_{wv}-H_v(p_w(r)+d-U/2)|+|\tau_{vv}-H_v(p_v(r)+d-U/2)|)\\
&\qquad\qquad + \left|\frac{2(H_v(p_w(r)+d-U/2)-H_v(p_v(r)+d-U/2))}{\vartheta+1}
-(p_w(r)-p_v(r))\right|\\
&\qquad \stackrel{\eqref{eq:uncertainty},\eqref{eq:drift}}{<} \vartheta U +
\frac{\vartheta-1}{\vartheta+1}\|\vec{p}(r)\|\,.\qedhere
\end{align*}
\end{proof}
We remark that if $(\vartheta-1) d < U$ and $U$ is known, it is beneficial to
refrain from having $v$ send a message to itself. Instead it estimates the
arrival time of the message using its hardware clock, yielding the following
corollary.
\begin{corollary}\label{cor:est}
Suppose $v\in C$ receives the pulse from $w\in C$ in round $r$ at a time from
$[H_v(t_v(r-1)),H_v(t_v(r-1))+\tau_1(r)+\tau_2(r)]$. Then
\begin{equation*}
\left|\frac{2(\tau_{wv}-H_v(p_v(r)))}{\vartheta+1}
-\left(d-\frac{U}{2}\right)-(p_w(r)-p_v(r))\right|<
\frac{\vartheta U}{2} + \frac{\vartheta-1}{\vartheta+1}(\|\vec{p}(r)\|+d)\,,
\end{equation*}
where $\tau_{wv}$ denotes the value of the respective variable in the algorithm
in round $r$.
\end{corollary}
\begin{proof}
By repeating the proof of \lemmaref{lem:est}, where the term
$|\tau_{vv}-H_v(p_v(r)+d-U/2)|$ gets replaced by
\begin{align*}
&\quad\left|H_v(p_v(r))+\frac{(\vartheta+1)(d-U/2)}{2}
-H_v\left(p_v(r)+d-\frac{U}{2}\right)\right|\\
&\leq \max\left\{\left|\frac{\vartheta+1}{2}-1\right|,
\left|\frac{\vartheta+1}{2}-\vartheta\right|\right\}\left(d-\frac{U}{2}\right)\\
&=
 \frac{\vartheta-1}{\vartheta+1}\left(d-\frac{U}{2}\right)\\
&<\frac{\vartheta-1}{\vartheta+1}\,d\,.\qedhere
\end{align*}
\end{proof}
In the sequel, we use the bounds provided by \lemmaref{lem:est}. However, the
reader should keep in mind that in case $(\vartheta-1)d\ll U$ and sufficiently
precise bounds on $U$ are known, \corollaryref{cor:est} shows how to effectively
cut the influence of the uncertainty in half.

Using \lemmaref{lem:est}, we can interpret the phase shifts $\Delta_v(r)$ as
outcomes of an approximate agreement step, yielding the following corollary.
\begin{corollary}\label{cor:step}
Suppose in round $r\in \N$, it holds for all $v,w\in C$ that $v$ receives the
pulse from $w\in C$ and itself in round $r$ during
$[H_v(t_v(r-1)),H_v(t_v(r-1))+\tau_1(r)+\tau_2(r)]$. Then
\begin{enumerate}
  \item $|\Delta_v(r)|< \vartheta(\|\vec{p}(r)\|+U)$ and
  \item $\max_{v,w\in C}\{p_v(r)-\Delta_v(r)-(p_w(r)-\Delta_w(r))\}\leq
  (5\vartheta-3)\|\vec{p}(r)\|/(2(\vartheta+1))+2\vartheta U$.
\end{enumerate}
\end{corollary}
\begin{proof}
By \lemmaref{lem:est}, we can interpret the values
$2(\tau_{wv}-\tau_{vv})/(\vartheta+1)$ as measurements of $p_w(r)-p_v(r)$ with
error $\delta=\vartheta U + (\vartheta-1)\|\vec{p}(r)\|/(\vartheta+1)$. Note
that shifting all values by $p_v(r)$ in an approximate agreement step changes
the result by exactly $p_v(r)$, implying that $p_v(r)-\Delta_v(r)$ equals the
result of an approximate agreement step with inputs $p_w(r)$, $w\in C$, and
error $\delta$ at node $v$. Thus, the claims follow from
\corollaryref{cor:max_correction} and \lemmaref{lem:convergence}, noting that
$1/2+2(\vartheta-1)/(\vartheta+1)= (5\vartheta-3)/(2(\vartheta+1))$.
\end{proof}

To derive a bound on $\|\vec{p}(r+1)\|$, it remains to analyze the effect of the
clock drift between the pulses. To this end, we examine how an established
timing relation between actions of two correct nodes deteriorates due to
measuring time using the inaccurate hardware clocks.
\begin{lemma}\label{lem:deteriorate}
Suppose $H_v(t_v')-H_v(t_v)=h_v\geq 0$ and $H_w(t_w')-H_v(t_w)=h_w\geq 0$. Then
\begin{equation*}
t_v-t_w+\frac{h_v}{\vartheta}-h_w\leq t_v'-t_w'\leq
t_v-t_w+h_v-\frac{h_w}{\vartheta}\,.
\end{equation*}
\end{lemma}
\begin{proof}
Since hardware clocks are increasing, $t_v'\geq t_v$ and $t_w'\geq t_w$. The
inequalities follow because hardware clock rates are between $1$ and
$\vartheta\geq 1$.
\end{proof}


This readily yields a bound on $\|\vec{p}(r+1)\|$ -- provided that all nodes can
compute when to send the next pulse on time.
\begin{corollary}\label{cor:iteration}
Assume that round $r\in \N$ is executed correctly. Then
\begin{equation*}
\|\vec{p}(r+1)\|\leq
\frac{2\vartheta^2+5\vartheta-5}{2(\vartheta+1)}\|\vec{p}(r)\|+(3\vartheta-1)U
+\left(1-\frac{1}{\vartheta}\right) T(r)\,.
\end{equation*}
\end{corollary}
\begin{proof}
For $v,w\in C$, assume w.l.o.g.\ that $p_v(r+1)-p_w(r+1)\geq 0$. By
\lemmaref{lem:deteriorate} and \corollaryref{cor:step}, we have that
\begin{align*}
&\quad~p_v(r+1)-p_w(r+1)\\
&\leq p_v(r)-p_w(r)+T(r)-\Delta_v(r)+\tau_1(r+1)-\tau_1(r)
-\frac{T(r)-\Delta_w(r)+\tau_1(r+1)-\tau_1(r)}{\vartheta}\\
&\leq p_v(r)-\Delta_v(r)-(p_w(r)-\Delta_w(r))
+\left(1-\frac{1}{\vartheta}\right)(T(r)+\tau_1(r+1)-\tau_1(r)+|\Delta_w(r)|)\\
&\leq \frac{2\vartheta^2+5\vartheta-5}{2(\vartheta+1)}\|\vec{p}(r)\|+(3\vartheta-1)U
+\left(1-\frac{1}{\vartheta}\right) (T(r)+\tau_1(r+1)-\tau_1(r))\,.\qedhere
\end{align*}
\end{proof}

This bound hinges on the assumption that the round is executed correctly. We
next establish sufficient conditions for this to be the case.

\begin{lemma}\label{lem:exec}
Suppose that
\begin{align*}
\tau_1(r)&\geq \vartheta (\|\vec{p}(r)\|-(d-U))\\
\tau_2(r)&\geq \vartheta (\|\vec{p}(r)\| + d)\\
T(r)&\geq \tau_1(r)+\tau_2(r)+\vartheta(\|\vec{p}(r)\|+U)\,.
\end{align*}
Then round $r$ is executed correctly.
\end{lemma}
\begin{proof}
Suppose $v,w\in C$. Denote by $t_{vw}\in
[p_v(r)+d-U, p_v(r)+d]$ the time when this message is received by $w$. We have that
\begin{equation*}
t_{vw}\geq p_v(r)+d-U\geq p_w(r)-\|\vec{p}(r)\|+d-U
\geq t_w(r-1)+\frac{\tau_1(r)}{\vartheta}-(\|\vec{p}(r)\|-(d-U))\geq t_w(r-1)\,,
\end{equation*}
showing that $H_w(t_{vw})\geq H_w(t_w(r-1))$, i.e., $w$ starts listening for the
pulse of $v$ on time. Similarly,
\begin{equation*}
t_{vw}\leq p_v(r)+d\leq p_w(r)+\|\vec{p}(r)\|+d\leq
p_w(r)+\frac{\tau_2(r)}{\vartheta}\,,
\end{equation*}
implying that $H_w(t_{vw})\leq H_w(p_w(r))+\tau_2(r) =
H_w(t_w(r-1))+\tau_1(r)+\tau_2(r)$. Thus, $w$ receives the pulse from $v$ before
it stops listening, and the first requirement of correct execution of round $r$
is met for all $v,w\in C$.

It remains to prove that for each $v\in C$, it holds that $T(r)-\Delta_v(r)\geq
\tau_1(r)+\tau_2(r)$. By the preconditions of the lemma, this is satisfied if
$\Delta_v(r)\leq \vartheta(\|\vec{p}(r)\|+U)$. As we already established the
precondition of \corollaryref{cor:step} for round $r$, the corollary shows that
this inequality is satisfied.
\end{proof}

We have almost all pieces in place to inductively bound $\|\vec{p}(r)\|$ and
determine suitable values for $\tau_1(r)$, $\tau_2(r)$, and $T(r)$. The last
missing bit is an anchor for the induction, i.e., a bound on $\|\vec{p}(1)\|$.
\begin{corollary}\label{cor:anchor}
$\|\vec{p}(1)\|\leq F+(1-1/\vartheta)\tau_1(1)=e(1)$.
\end{corollary}
\begin{proof}
Since $H_v(0)\in [0,F)$ for all $v\in C$, $t_v(0)\in [0,F)$ for all $v\in C$.
The claim follows by applying \lemmaref{lem:deteriorate}.
\end{proof}

\begin{theorem}\label{thm:basic}
Suppose that \conditionref{cond:constraints} is satisfied. Then, for all $r\in
\N$, it holds that $\|\vec{p}(r)\|\leq e(r)$. If $\alpha=
(6\vartheta^2+5\vartheta-9)/(2(\vartheta+1)(2-\vartheta))<1$ (which holds for
$\vartheta\leq 1.1$), we can choose the parameters such that the condition holds
and \algref{alg:basic} has steady state error
\begin{equation*}
E = \lim_{r\to \infty} e(r) \leq
\frac{(\vartheta-1)d+(4\vartheta-2)U}{(2-\vartheta)\alpha}\,.
\end{equation*}
\end{theorem}
\begin{proof}
To show the first part, inductively use \lemmaref{lem:exec} and
\lemmaref{cor:iteration} to show that round $r$ is executed correctly and that
$\|\vec{p}(r+1)\|\leq e(r+1)$, respectively; the induction anchor is given by
$\|\vec{p}(1)\|\leq e(1)$ according to \corollaryref{cor:anchor}. The second
part directly follows from \lemmaref{lem:solvable}.
\end{proof}

\section{Phase and Frequency Synchronization Algorithm}\label{sec:frequency}

In this section, we extend the phase synchronization algorithm to also
synchronize frequencies. The basic idea is to apply the approximate agreement
not only to phase offsets, but also to frequency offsets. To this end, in each
round the phase difference is measured twice, applying any phase correction only
after the second measurement. This enables nodes to obtain an estimate of the
relative clock speeds, which in turn is used to obtain an estimate of the
differences in clock speeds.

Ensuring that this procedure is executed correctly is straightforward by
limiting $|\mu_v(r)-1|$ to be small, where $\mu_v(r)$ is the factor by which
node $v$ changes its clock rate during round $r$. However, constraining this
multiplier means that approximate agreement steps cannot be performed correctly
in case $\mu_v(r+1)$ would lie outside the valid range of multipliers. This is
fixed by introducing a correction that ``pulls'' frequencies back to the default
rate.

Of course, for all this to be meaningful, we need to assume that hardware clock
rates do not change faster than the algorithm can adjust the multipliers to keep
the effective frequencies aligned.

\subsection{Additional Assumptions on the Clocks}

We require that clock rates satisfy a Lipschitz condition as well. In the
following, we assume that $H_v$ is differentiable (for all $v\in C$) with
derivative $h_v$, where $h_v$ satisfies for $t,t\in \R^+_0$ that
\begin{equation}\label{eq:acc}
|h_v(t')-h_v(t)|\leq \nu |t'-t|
\end{equation}
for some $\nu>0$. Note that we maintain the model assumption that hardware clock
rates are close to $1$ at all times, i.e., $1\leq h_v(t)\leq \vartheta$ for all
$t\in \R^+_0$.

\subsection{Algorithm}

\begin{algorithm}[t!]\label{alg:frequency}
\caption{Phase and frequency synchronization algorithm, code for node~$v\in C$.
Time $t_v(r)$, $r\in \N_0$, is the time when round $r+1$ starts.}
// $H_w(0)\in [0,F)$ for all $w\in V$\\
wait until time $t_v(0)$ with $H_v(t_v(0))=F$\;
// initialize clock rate multiplier\\
$\mu_v(0):=\mu_v(1):=\vartheta$\;
\ForEach{round $r \in \N$} {
	// phase correction step\\
	start listening for messages\;
	wait until local time $H_v(t_v(r-1))+\tau_1/\mu_v(r-1)$\;
	broadcast clock pulse to all nodes (including self)\; 
	wait until local time $H_v(t_v(r-1))+(\tau_1+\tau_2)/\mu_v(r)$\;
	\For{each node $w\in V$}{
		$\tau_{wv}:= H_v(t_{wv})$ (first message from $w$ while listening at time
		$t_{wv}$; $\tau_{wv}:=\infty$ if none)\;
	}
	$S_{v} \gets \{2(\tau_{wv}-\tau_{vv})/(\vartheta+1)\mid w\in
	V\}$ (as multiset)\;
	let $S_v^k$ denote the $k^{th}$ smallest element of $S_v$\; 
	$\Delta_{v}(r) \gets \dfrac{ S_{v}^{f + 1} + S_{v}^{n - f}}{2}$\;
	// frequency correction step\\
	start listening for messages\;
	wait until local time
	$H_{v}(t_{v}(r-1))+(\tau_1+\tau_2+\tau_3)/\mu_v(r)$\;
	broadcast clock pulse to all nodes (including self)\; 
	wait until local time $H_v(t_v(r-1))+(\tau_1+\tau_2+\tau_3+\tau_4)/\mu_v(r)$\;
	\For{each node $w\in V$}{
		$\tau_{wv}':= H_v(t_{wv}')$ (first message from $w$ while listening at time
		$t_{wv}'$; $\tau_{wv}:=\infty$ if none)\;
		$\Delta_{wv}:= H_v(t_{wv}')-H_v(t_{wv})$\;
	}
	$S_{v} \gets \{1-\mu_v(r)\Delta_{wv}/(\tau_2+\tau_3)\mid w\in V\}$ (as
	multiset)\;
	let $S_v^k$ denote the $k^{th}$ smallest element of $S_v$\; 
	$\xi_{v}(r) \gets \dfrac{ S_{v}^{f + 1} + S_{v}^{n - f}}{2}$\;
	$\hat{\mu}_v(r+1)\gets \mu_v(r)+2\xi_v(r)/(\vartheta+1)$\;
		// pull back towards nominal frequency by $\varepsilon$, ensure minimum
		and maximum rate\\
	\If{$\hat{\mu}_v(r+1)\leq \vartheta$}{
		$\mu_v(r+1)\gets\max\{\hat{\mu}_v(r+1)+\varepsilon,1\}$;
	}
	\Else{
		$\mu_v(r+1)\gets \min\{\hat{\mu}_v(r+1)-\varepsilon,\vartheta^2\}$\;
	}
	wait until time $t_v(r)$ with $H_v(t_v(r))+(T-\Delta_v(r))/\mu_v(r)$;~~//
	nominal round length is $T$\\
}
\end{algorithm}

\algref{alg:frequency} gives the pseudocode of our approach. Mostly, the
algorithm can be seen as a variant of \algref{alg:basic} that allows for
speeding up clocks by factors $\mu_v(r)\in [1,\vartheta^2]$, where $\vartheta
h_v(t)$ is considered the nominal rate at time $t$.\footnote{Given that
hardware clock speeds may differ by at most factor $\vartheta$, nodes need to
be able to increase or decrease their rates by factor $\vartheta$: a single
deviating node may be considered faulty by the algorithm, so each node must be
able to bridge this speed difference on its own.} For simplicity, we fix all
local waiting times independently of the round length.

The main difference to \algref{alg:basic} is that a second pulse signal is sent
before the phase correction is applied, enabling to determine the rate
multipliers for the next round by an approximate agreement step as well. A
frequency measurement is obtained by comparing the (observed) relative rate of
the clock of node $w$ during a local time interval of length $\tau_2+\tau_3$ to
the desired relative clock rate of $1$. Since the clock of node $v$ is
considered to run at speed $\mu_v(r)h_v(t)$ during the measurement period, the
former takes the form $\mu_v(r)\Delta_{wv}/(\tau_2+\tau_3)$, where $\Delta_{wv}$
is the time difference between the arrival times of the two pulses from $w$
measured with $H_v$. The approximate agreement step results in a new multiplier
$\hat{\mu}_v(r+1)$ at node $v$; we then move this result by $\varepsilon$ in
direction of the nominal rate multiplier $\vartheta$ and ensure that we remain
within the acceptable multiplier range $[1,\vartheta^2]$.

To fully specify the algorithm, we need to determine how long the waiting
periods are (in terms of local time) and choose $\varepsilon$. Here, we must
ensure for all $r\in \N$ that
\begin{enumerate}
  \item for all $v,w\in C$, the message $v$ broadcasts at time
  $t_v(r-1)+\tau_1/\mu_v(r-1)$ is received by $w$ at a local time from
  $[H_w(t_w(r-1)),H_w(t_w(r-1))+\tau_1/\mu_v(r-1)+\tau_2/\mu_w(r)]$,
  \item for all $v,w\in C$, the message $v$ broadcasts at time
  $t_v(r-1)+\tau_1/\mu_v(r-1)+(\tau_2+\tau_3)/\mu_v(r)$ is received by $w$ at a
  local time from $[H_w(t_w(r-1))+\tau_1/\mu_v(r-1)+\tau_2/\mu_w(r),
  H_w(t_w(r-1))+\tau_1/\mu_v(r-1)+(\tau_2+\tau_3+\tau_4)/\mu_w(r)]$, and
  \item for all $v\in C$, $T-\Delta_v(r)\geq
  \tau_1/\mu_v(r-1)+(\tau_2+\tau_3+\tau_4)/\mu_v(r)$, i.e., $v$ computes
  $H_v(t_v(r))$ \emph{before} time $t_v(r)$.
\end{enumerate}
If these conditions are satisfied for $r\in \N$, we say that \emph{round $r$
was executed correctly.}

We now specify the constraints our choices for the parameters must satisfy to
ensure that all rounds are executed correctly and both phase and frequency
errors converge to small values.
\begin{condition}\label{cond:freq_constraints}
Set $\bar{\vartheta}:=\vartheta^3$. Define 
\begin{equation*}
e(1):=\max\left\{F+\left(1-\frac{1}{\bar{\vartheta}}\right)\tau_1,
\frac{(1-1/\bar{\vartheta})T+(3\bar{\vartheta}-1)U}{1-\bar{\beta}}\right\}
\end{equation*}
and, inductively for $r\in \N$,
\begin{equation*}
e(r+1):=\frac{2\bar{\vartheta}^2+5\bar{\vartheta}-5}{2(\bar{\vartheta}+1)}\,e(r)
+(3\bar{\vartheta}-1)U+\left(1-\frac{1}{\vartheta}\right)T\,.
\end{equation*}
We require that
\begin{align*}
\tau_1&\geq \bar{\vartheta} e(1)\\
\tau_2&\geq \bar{\vartheta}(e(1)+d)\\
\tau_3&\geq
\bar{\vartheta}\left(e(1)+\left(1-\frac{1}{\bar{\vartheta}}\right)(\tau_1+\tau_2)\right)\\
\tau_4&\geq
\bar{\vartheta}\left(e(1)+d+\left(1-\frac{1}{\bar{\vartheta}}\right)(\tau_1+\tau_2)\right)\\
T&\geq \tau_1+\tau_2+\tau_3+\tau_4+\bar{\vartheta}(e(1)+U)\\
\varepsilon&\geq 2\left((\vartheta-1)(\vartheta^3-1)+
2\vartheta^3\left(1-\frac{1}{\vartheta^3}\right)^2
+\frac{2\vartheta^3 U}{\tau_2+\tau_3}+2(\vartheta^3+1)\nu T\right)\,.
\end{align*}
\end{condition}
Here, all but the last conditions mimic \conditionref{cond:constraints}, where
the bounds on $\tau_3$ and $\tau_4$ account for the fact that between the first
and the second pulse of each round, the nodes' opinion on the ``synchronized
time'' drift apart slowly. The lower bound on $\varepsilon$ ensures that the
pull-back of multipliers to the nominal ones is sufficiently strong to guarantee
that, in fact, multipliers will never leave the valid range of
$[1,\vartheta^2]$. We now show that these constraints can be satisfied provided
that $\vartheta$ is not too large.

\begin{lemma}\label{lem:solvable_freq}
\conditionref{cond:freq_constraints} can be satisfied such that $\lim_{r\to
\infty} e(r)<\infty$ if
\begin{equation*}
\bar{\alpha}:=\bar{\beta}+(4\bar{\vartheta}+3)(\bar{\vartheta}-1)<1\,,
\end{equation*}
where $\bar{\beta}:=(2\bar{\vartheta}^2+5\bar{\vartheta}-5)/(2(\bar{\vartheta}+1))$. Here, we may
choose any $T\geq T_0\in \BO(F+d+U)$. In this case,
\begin{equation*}
\lim_{r\to \infty} e(r) = \frac{(1-1/\bar{\vartheta})T+(3\bar{\vartheta}
-1)U}{1-\bar{\beta}}\,.
\end{equation*}
\end{lemma}
\begin{proof}
We choose $\tau_1$, $\tau_2$, $\tau_3$, and $\tau_4$ minimal such that the
respective constraints are satisfied, and pick any feasible $\varepsilon$.
Hence, the remaining constraints are that
\begin{equation}\label{eq:T}
T\geq \bar{\vartheta}((4\bar{\vartheta}+3)e(1)+(2\bar{\vartheta}+1)d+U)
\end{equation}
and
\begin{equation*}
e(1)=\max\left\{F+\left(1+\frac{1}{\bar{\vartheta}}\right)e(1),
\frac{(1-1/\bar{\vartheta})T+(3\bar{\vartheta}-1)U}{1-\bar{\beta}}\right\}.
\end{equation*}
Using that $2-\bar{\vartheta}>0$ (which is a weaker constraint than $\bar{\alpha}<1$),
assuming that $e(1)$ equals the first term of the maximum would yield that
\begin{equation*}
e(1)= \frac{F}{2-\bar{\vartheta}}\,,
\end{equation*}
and clearly there is a $T_0\in \BO(F+d+U)$ such that \eqref{eq:T} is satisfied
for any $T\geq T_0$. Assuming that $e(1)$ equals the second term in the maximum,
\eqref{eq:T} becomes
\begin{equation*}
T\geq \bar{\vartheta}\left((4\bar{\vartheta}+3)
\left(\frac{(1-1/\bar{\vartheta})T+(3\bar{\vartheta}-1)U}{1-\bar{\beta}}\right)
+(2\bar{\vartheta}+1)d+U)\right).
\end{equation*}
Using that $\bar{\alpha}<1$, we can resolve this to
\begin{equation*}
T\geq \bar{\vartheta}\cdot \frac{(4\bar{\vartheta}+3)(3\bar{\vartheta}+1)U+(1+\bar{\beta})
((2\bar{\vartheta}+1)d+U)}{1-\bar{\alpha}}\in\BO(U+d)\,.
\end{equation*}
For the final claim, observe that by induction on $r$, we have that
\begin{align*}
\lim_{r\to \infty}e(r)
&=\lim_{r\to \infty}\left(\bar{\beta}^{r-1}e(1)
+\sum_{i=1}^{r-1}\bar{\beta}^{i-1}
\left((3\bar{\vartheta}-1)U+\left(1-\frac{1}{\vartheta}\right)T\right)\right)\\
&= \frac{(1-1/\bar{\vartheta})T+(3\bar{\vartheta}-1)U}{1-\bar{\beta}}\,.\qedhere
\end{align*}
\end{proof}

\subsection{Analysis}\label{sec:freq_analysis}

In the following, denote by $\vec{p}(r)$ and $\vec{q}(r)$, $r\in \N$, the
vectors of times when nodes $v\in C$ broadcast their first and second pulse in
round $r$, respectively. Thus, we have that
$H_v(p_v(r))=H_v(t_v(r-1))+\tau_1/\mu_v(r-1)$
and $H_v(q_v(r))=H_v(t_v(r-1))+\tau_1/\mu_v(r-1)+(\tau_2+\tau_3)/\mu_v(r)$.

We will first make use of the analysis we performed for the phase correction
algorithm to show that all rounds are executed correctly. Then we will refine
the analysis by examining the impact of the frequency correction steps.

\subsubsection*{Phase Correction Steps}

Observe that because for all $r\in \N_0$ and $v\in C$, we have that $1\leq
\mu_v(r)\leq \vartheta^2$, for all times $t$ we have that $1\leq
\mu_v(r)h_v(t)\leq \vartheta^3=\bar{\vartheta}$. Thus, we may interpret the waiting
periods of \algref{alg:frequency} as nodes waiting for $\tau_1$, $\tau_2$, etc.\
local time with hardware clocks of drift $\bar{\vartheta}=\vartheta^3$. Thus, we can
make use of the same arguments as in \sectionref{sec:basic_analysis} to obtain a
series of results.
\begin{corollary}
For all $r\in \N$, $\|\vec{q}(r)\|\leq
\|\vec{p}(r)\|+(1-1/\bar{\vartheta})(\tau_1+\tau_2)$.
\end{corollary}
\begin{proof}
By application of \lemmaref{lem:deteriorate}.
\end{proof}
\begin{corollary}\label{cor:exec}
Suppose that
\begin{align*}
\tau_1&\geq \vartheta (\|\vec{p}(r)\|-(d-U))\\
\tau_2&\geq \vartheta (\|\vec{p}(r)\| + d)\\
\tau_3&\geq \vartheta (\|\vec{q}(r)\|-(d-U))\\
\tau_4&\geq \vartheta (\|\vec{q}(r)\| + d)\\
T&\geq \tau_1+\tau_2+\tau_3+\tau_4+\vartheta(\|\vec{p}(r)\|+U)\,.
\end{align*}
Then round $r$ is executed correctly.
\end{corollary}
\begin{proof}
As for \lemmaref{lem:exec}, where the pulse in the frequency correction step is
analyzed analogously.
\end{proof}
\begin{theorem}\label{thm:basic_freq}
Suppose that \conditionref{cond:freq_constraints} is satisfied and that
\begin{equation*}
\bar{\alpha}:=\bar{\beta}+(4\bar{\vartheta}+3)(\bar{\vartheta}-1)<1\,,
\end{equation*}
where $\bar{\beta}:=(2\bar{\vartheta}^2+5\bar{\vartheta}-5)/(2(\bar{\vartheta}+1))$ (this is the
case for $\vartheta\leq 1.011$). Then, for all $r\in \N$, it holds that
$\|\vec{p}(r)\|\leq e(r)$ and the algorithm has steady state error
\begin{equation*}
E\leq \frac{(1-1/\bar{\vartheta})T+(3\bar{\vartheta}-1)U}{1-\bar{\beta}}\,.
\end{equation*}
In particular, all rounds $r\in \N$ are executed correctly.
\end{theorem}
\begin{proof}
As for \theoremref{thm:basic}, where we replace $\vartheta$ with $\bar{\vartheta}$,
\lemmaref{lem:exec} with \corollaryref{cor:exec} and \lemmaref{lem:solvable}
with \lemmaref{lem:solvable_freq}. However, the induction step requires that
we can apply \lemmaref{lem:exec} again in step $r+1$ if we could do so in step
$r\in \N$. This readily follows from \conditionref{cond:freq_constraints} if
$e(r+1)\leq e(r)$ for all $r\in \N$. 

We show this by induction on $r$. Abbreviate
$x:=(3\bar{\vartheta}-1)U+(1-1/\bar{\vartheta})T$. Our claim is that (i) for $r\in \N$,
$e(r)\geq x/(1-\bar{\beta})$ and (ii) for $r\geq 2$, $e(r)\leq e(r-1)$. The base case
$r=1$ requires (i) only, which holds by definition of $e(1)$. For the
step from $r$ to $r+1$, we bound
\begin{equation*}
e(r+1)=\bar{\beta}e(r)+x\geq \frac{\bar{\beta} x}{1-\bar{\beta}}+x=\frac{x}{1-\bar{\beta}}
\end{equation*}
and
\begin{equation*}
e(r)-e(r+1)=(1-\bar{\beta})e(r)-x\geq x-x=0\,.
\end{equation*}
Finally, observe that our reasoning shows as part of the inductive argument that
all rounds are executed correctly.
\end{proof}

\subsubsection*{Frequency Correction Steps}

In the following, we assume that the prerequisites of
\theoremref{thm:basic_freq} are satisfied. In particular, all rounds are
executed correctly, i.e., we can assume that correct nodes receive each others'
pulses. We introduce some notation to capture the behavior of the (logical)
rates of the nodes' clocks. This notation may seem somewhat cumbersome;
basically, the reader may think of the clock rates $h_v(t)$ being almost
constant, implying that all considered values for a given node $v\in C$ are
essentially the same, slowly deviating at rate at most $\nu$.

By $\vec{\rho}(r)$, we denote the vector whose entries are the intervals of
clock rate ranges of nodes $v\in C$ between the first pulses in rounds $r\in \N$
and $r+1$. Concretely,
\begin{equation*}
\vec{\rho}(r)_v:=\left[
\min_{p_v(r)\leq t\leq p_v(r+1)}\{\mu_v(r)h_v(t)\},
\max_{p_v(r)\leq t\leq p_v(r+1)}\{\mu_v(r)h_v(t)\}
\right].
\end{equation*}
By $\|\vec{\rho}(r)\|$, we denote the difference between maximum and minimum
rate in $\vec{\rho}(r)$, i.e., 
\begin{equation*}
\|\vec{\rho}(r)\|:=\max_{v\in C}\max_{p_v(r)\leq t\leq
p_v(r+1)}\{\mu_v(r)h_v(t)\} -\min_{v\in C}\min_{p_v(r)\leq t\leq p_v(r+1)}\{\mu_v(r)h_v(t)\}\,.
\end{equation*}
Furthermore, we denote by $\bar{\rho}(r)_v:=\mu_v(r)h_v((p_v(r)+p_v(r+1))/2)$,
by $\bar{\rho}(r)$ the respective vector, and by $\|\bar{\rho}(r)\|:=\max_{v\in
C}\{\bar{\rho}(r)\}-\min_{v\in C}\{\bar{\rho}(r)\}$. Note that
$\bar{\rho}(r)_v\in \vec{\rho}(r)_v$ by definition.

We start by showing that $\bar{\rho}(r)_v$ approximates $\mu_v(r)h_v(t)$ well
for times $t$ between pulse $r$ and $r+1$ of $v\in C$, i.e., we may see
$\bar{\rho}(r)_v$ as ``the'' clock rate of $v$ in round $r$.
\begin{lemma}\label{lem:stability}
Let $t\in [p_v(r),p_v(r+1)]$ for some $v\in C$ and $r\in \N$. Then
\begin{equation*}
|\mu_v(r)h_v(t)-\bar{\rho}(r)_v|<\nu\, \frac{T+\tau_2}{2}\,.
\end{equation*}
\end{lemma}
\begin{proof}
Using that hardware clock rates are at least $1$ and that
$|\Delta_v(r)|<\max\{\tau_1,\tau_2\}=\tau_2$, we see that
\begin{equation*}
\left|t-\frac{p_v(r+1)+p_v(r)}{2}\right|\leq \frac{|p_v(r+1)-p_v(r)|}{2}
\leq \frac{|T-\Delta_v(r)|}{2\mu_v(r)}<\frac{T+\tau_2}{2\mu_v(r)}\,.
\end{equation*}
By our assumptions on the hardware clocks, this yields that
\begin{equation*}
\left|\mu_v(r)\left(h_v(t)-h_v\left(\frac{p_v(r+1)+p_v(r)}{2}\right)\right)\right|
\leq \mu_v(r)\cdot\nu \left|t-\frac{p_v(r+1)+p_v(r)}{2}\right|<\nu
\,\frac{T+\tau_2}{2}\,.\qedhere
\end{equation*}
\end{proof}
Two corollaries relate the progress of the hardware clocks between (i) $p_v(r)$
and $q_v(r)$ and (ii) $t_{wv}'$ and $t_{wv}$ to $\bar{\rho}(r)_v$, respectively.
\begin{corollary}\label{cor:stability}
For $v\in C$ and $r\in \N$, we have that
\begin{equation*}
|\bar{\rho}(r)_v(q_v(r)-p_v(r))-(\tau_2+\tau_3)|<\nu T(\tau_2+\tau_3)\,.
\end{equation*}
\end{corollary}
\begin{proof}
Let $\rho\in \vec{\rho}(r)_v$ such that
$\rho(q_v(r)-p_v(r))=\tau_2+\tau_3$. By definition of
$\vec{\rho}(r)_v$ and the mean value theorem, such a $\rho$ exists and $\rho=\mu_v(r)h_v(t)$ for
some $t\in [p_v(r),p_v(r+1)]$. By \lemmaref{lem:stability},
$|\rho-\bar{\rho}(r)_v|<\nu T$. Thus,
\begin{equation*}
|\bar{\rho}(r)_v(q_v(r)-p_v(r))-(\tau_2+\tau_3)|
=|\rho-\bar{\rho}(r)_v|(q_v(r)-p_v(r))
=|\rho-\bar{\rho}(r)_v|\,\frac{\tau_2+\tau_3}{\rho}\,
<\nu T(\tau_2+\tau_3)\,.\qedhere
\end{equation*}
\end{proof}

\begin{corollary}\label{cor:stability2}
For $v,w\in C$ and $r\in \N$, we have that
\begin{equation*}
|\mu_v(r)(H_v(t_{wv}')-H_v(t_{wv}))-\bar{\rho}(r)_v(t_{wv}'-t_{wv})|
<\nu T(\tau_2+\tau_3)\,.
\end{equation*}
\end{corollary}
\begin{proof}
Let $\bar{\rho}\in \vec{\rho}(r)_v$ such that
$t_{wv}'-t_{wv}=\mu_v(r)(H_v(t_{wv}')-H_v(t_{wv})$. By definition of
$\vec{\rho}(r)_v$ and the mean value theorem, such a $\rho$ exists and
$\rho=\mu_v(r)h_v(t)$ for some $t\in [t_{wv},t_{wv}']\subseteq
[p_v(r),p_v(r+1)]$. By \lemmaref{lem:stability}, $|\rho-\bar{\rho}(r)_v|<\nu
(T+\tau_2)/2$.
Thus,
\begin{align*}
|\mu_v(r)(H_v(t_{wv}')-H_v(t_{wv}))-\bar{\rho}(r)_v(t_{wv}'-t_{wv})|
&=|\rho-\bar{\rho}(r)_v|(t_{wv}'-t_{wv})\\
&<\nu \,\frac{T+\tau_2}{2}(\tau_2+\tau_3+U)\\
&<\nu T(\tau_2+\tau_3)\,,
\end{align*}
where the second last step exploits that $t_{wv}'-t_{wv}\leq
q_w(r)+d-(p_w(r)+d-U)\leq \tau_2+\tau_3+U$, since clock rates are at least $1$,
and the final inequality easily follows from
\conditionref{cond:freq_constraints}.
\end{proof}

These results put us in the position to prove that
$1-\mu_v(r)\Delta_{wv}/(\tau_2+\tau_3)$ is indeed a good estimate of
$\bar{\rho}(r)_w-\bar{\rho}(r)_v$. Thus, this (computable) value can serve as a
proxy for the difference between ``the'' clock rates of $w$ and $v$ in round $r$.
\begin{lemma}\label{lem:freq_est}
For $v,w\in C$ and $r\in \N$, we have that
\begin{equation*}
\left|\bar{\rho}(r)_w-\bar{\rho}(r)_v
-\left(1-\frac{\mu_v(r)\Delta_{wv}}{\tau_2+\tau_3}\right)\right|
\leq \vartheta^3\left(1-\frac{1}{\vartheta^3}\right)^2
+ \frac{\vartheta^3 U}{\tau_2+\tau_3}+(\vartheta^3+1) \nu T\,.
\end{equation*}
\end{lemma}
\begin{proof}
We have
\begin{equation}\label{eq:freq_est_1}
|t_{wv}'-t_{wv}-(q_w(r)-p_w(r))|\leq U
\end{equation}
and by Corollaries~\ref{cor:stability} and~\ref{cor:stability2} that
\begin{align}
\left|\frac{q_w(r)-p_w(r)}{\tau_2+\tau_3}-\frac{1}{\bar{\rho}(r)_w}\right|&
<\frac{\nu T}{\bar{\rho}(r)_w}\leq \nu T\label{eq:freq_est_2}\\
\left|\frac{\mu_v(r)\Delta_{wv}}{t_{wv}'-t_{wv}}-\bar{\rho}(r)_v\right|
&<\nu T\,.\label{eq:freq_est_3}
\end{align}
Note that $|\mu_v(r)\Delta_{wv}/(t_{wv}'-t_{wv})|\leq \vartheta^3$. Therefore,
\begin{align*}
\left|\frac{\bar{\rho}(r)_v}{\bar{\rho}(r)_w}
-\frac{\mu_v(r)\Delta_{wv}}{\tau_2+\tau_3}\right|
&= \left|\frac{\bar{\rho}(r)_v}{\bar{\rho}(r)_w}
-\frac{\mu_v(r)\Delta_{wv}}{t_{wv}'-t_{wv}}
\cdot\frac{t_{wv}'-t_{wv}}{q_w(r)-p_w(r)}
\cdot\frac{q_w(r)-p_w(r)}{\tau_2+\tau_3}\right|\\
&\stackrel{\eqref{eq:freq_est_1}}{\leq}
\left|\frac{\bar{\rho}(r)_v}{\bar{\rho}(r)_w}
-\frac{\mu_v(r)\Delta_{wv}}{t_{wv}'-t_{wv}}
\cdot\frac{q_w(r)-p_w(r)}{\tau_2+\tau_3}\right|+
\frac{\vartheta^3 U}{\tau_2+\tau_3}\\
&\stackrel{\eqref{eq:freq_est_2}}{\leq}
\left|\frac{\bar{\rho}(r)_v}{\bar{\rho}(r)_w}
-\frac{\mu_v(r)\Delta_{wv}}{t_{wv}'-t_{wv}}
\cdot \frac{1}{\bar{\rho}(r)_w}\right|+
\frac{\vartheta^3 U}{\tau_2+\tau_3}+\vartheta^3 \nu T\\
&\stackrel{\eqref{eq:freq_est_3}}{\leq}
\frac{\vartheta^3 U}{\tau_2+\tau_3}+(\vartheta^3+1) \nu T\,.
\end{align*}
Moreover,
\begin{equation*}
\left|\bar{\rho}(r)_w-\bar{\rho}(r)_v
-\left(1-\frac{\bar{\rho}(r)_v}{\bar{\rho}(r)_w}\right)\right|
=\left(1-\frac{1}{\bar{\rho}(r)_w}\right)|\bar{\rho}(r)_w-\bar{\rho}(r)_v|
\leq \left(1-\frac{1}{\vartheta^3}\right)(\vartheta^3-1)\,.
\end{equation*}
We conclude that
\begin{equation*}
\left|\bar{\rho}(r)_w-\bar{\rho}(r)_v
-\left(1-\frac{\mu_v(r)\Delta_{wv}}{\tau_2+\tau_3}\right)\right|
\leq \vartheta^3\left(1-\frac{1}{\vartheta^3}\right)^2
+ \frac{\vartheta^3 U}{\tau_2+\tau_3}+(\vartheta^3+1) \nu T\,.\qedhere
\end{equation*}
\end{proof}
We remark that the $\Theta((1-1/\vartheta^3)^2)$ factor is, more precisely,
bounded as $\Theta((1-1/\vartheta^3)\|\bar{\rho}(r)\|)$. However, for this to be
of use, we would have to choose $\varepsilon$ depending on $r$. Since
rule-of-thumb calculations show that this term is unlikely to be significant in
any real system and the improvement would not extend to the self-stabilizing
variant of the algorithm, we refrained from adding this additional complication.

Given that we can bound the ``measurement error'' of the frequency correction
step by \lemmaref{lem:freq_est}, the results from \sectionref{sec:approx} can be
invoked to show convergence. First, we analyze the properties of
$\hat{\mu}_v(r+1)$, which \lemmaref{lem:freq_convergence} then uses to control
$\mu_v(r+1)$.
\begin{lemma}\label{lem:freq_step}
For $v\in C$ and $r\in \N$, abbreviate
$\bar{t}_v:= (p_v(r)+p_v(r+1))/2$, i.e.,
$\bar{\rho}(r)_v=\mu_v(r)h_v(\bar{t}_v)$. Then, for all $v,w\in C$,
\begin{equation*}
|\hat{\mu}_v(r+1)h_v(\bar{t}_v)-\hat{\mu}_w(r+1)h_w(\bar{t}_w)|\leq
\frac{2\vartheta-1}{2}\,\|\bar{\rho}(r)\|+\vartheta\varepsilon\,.
\end{equation*}
Furthermore,
\begin{align*}
(\hat{\mu}_v(r+1)-\varepsilon)h_v(\bar{t}_v)
&\leq \max_{u\in C}\left\{\mu_u(r)h_u(\bar{t}_u)\right\}-\frac{\varepsilon}{2}\\
(\hat{\mu}_v(r+1)+\varepsilon)h_v(\bar{t}_v)
&\geq \min_{u\in
C}\left\{\mu_u(r)h_u(\bar{t}_u)\right\}+\frac{\varepsilon}{2}\,.
\end{align*}
\end{lemma}
\begin{proof}
Set $\delta:=\vartheta^3(1-\vartheta^{-3})^2 + \vartheta^3
U/(\tau_2+\tau_3)+(\vartheta^3+1) \nu T$. Observe that, by
\lemmaref{lem:freq_est}, we can interpret $\bar{\rho}(r)_v+\xi_v(r)$, $v\in C$,
as the results of an approximate agreement step with error $\delta$ on inputs
$\bar{\rho}(r)$. By \lemmaref{lem:convergence}, this implies that
\begin{equation*}
|\hat{\mu}_v(r)h_v(\bar{t}_v)+\xi_v(r)-(\hat{\mu}_w(r)h_v(\bar{t}_w)+\xi_w(r))|
\leq \frac{\|\bar{\rho}(r)\|}{2}+2\delta\,.
\end{equation*}
By \corollaryref{cor:max_correction}, $\max_{u\in C}|\{\xi_u(r)|\}\leq
\|\bar{\rho}(r)\|+\delta$. Hence, we have for $u\in C$ that
\begin{equation}\label{eq:approx}
|\hat{\mu}_u(r+1)h_u(\bar{t}_u)-(\hat{\mu}_u(r)h_u(\bar{t}_u)+\xi_u(r))|
=\left|\frac{2h_u(\bar{t}_u)}{\vartheta+1}-1\right|\cdot|\xi_u(r)|\leq
\frac{\vartheta-1}{\vartheta+1}(\|\bar{\rho}(r)\|+\delta)\,.
\end{equation}
Using this bound for both $v$ and $w$, we conclude that
\begin{align*}
|\hat{\mu}_v(r+1)h_v(\bar{t}_v)-\hat{\mu}_w(r+1)h_w(\bar{t}_w)|
&\leq \frac{\|\bar{\rho}(r)\|}{2}+2\delta+
\frac{2(\vartheta-1)}{\vartheta+1}(\|\bar{\rho}(r)\|+\delta)\\
&<\frac{2\vartheta-1}{2}\,\|\bar{\rho}(r)\|+(\vartheta+1) \delta\\
&<\frac{2\vartheta-1}{2}\,\|\bar{\rho}(r)\|+\vartheta\varepsilon\,.
\end{align*}
For the second claim of the lemma, we apply \lemmaref{lem:validity}. Together
with \eqref{eq:approx}, this shows for $v\in C$ that
\begin{align*}
\hat{\mu}_v(r+1)h_v(\bar{t}_v)
&< \max_{u\in C}\left\{\mu_u(r)h_u(\bar{t}_u)\right\}+\delta
+\frac{h_v(\bar{t}_v)-1}{2}\,(\|\bar{\rho}(r)\|+\delta)\\
\hat{\mu}_v(r+1)h_v(\bar{t}_v)
&> \min_{u\in
C}\left\{\mu_u(r)h_u(\bar{t}_u)\right\}-\left(\delta
+\frac{h_v(\bar{t}_v)-1}{2}\,(\|\bar{\rho}(r)\|+\delta)\right),
\end{align*}
where we used that $2h_v(\bar{t}_v)/(\vartheta+1)-1\leq (h_v(\bar{t}_v)-1)/2$.
By \conditionref{cond:freq_constraints} (and because $\|\bar{\rho}(r)\|\leq
\vartheta^3-1$),
\begin{equation*}
\frac{\varepsilon}{2}\, h_v(\bar{t}_v)\geq
\left(\delta+\frac{(\vartheta-1)(\vartheta^3-1)}{2}\right) h_v(\bar{t}_v) >
\delta +\frac{h_v(\bar{t}_v)-1}{2}\,(\|\bar{\rho}(r)\|+\delta)\,.
\end{equation*}
Combining this with the above inequalities completes the proof.
\end{proof}

\begin{lemma}\label{lem:freq_convergence}
For round $r\in \N$ and $v\in C$, abbreviate $\bar{t}_v:= (p_v(r)+p_v(r+1))/2$,
i.e., $\bar{\rho}(r)_v=\mu_v(r)h_v(\bar{t}_v)$. For all $v,w\in C$, we have that
\begin{equation*}
|\mu_v(r+1)h_v(\bar{t}_v)-\mu_w(r+1)h_w(\bar{t}_w)|\leq \max\left\{
\frac{2\vartheta-1}{2}\,\|\bar{\rho}(r)\|
+ 3\vartheta \varepsilon, \|\bar{\rho}(r)\|-\frac{\varepsilon}{2}\right\}.
\end{equation*}
\end{lemma}
\begin{proof}
Let $v\in C$ and $w\in C$ maximize and minimize $\mu_u(r+1)h_u(\bar{t}_u)$ over
$u\in C$, respectively. By \lemmaref{lem:freq_step}, we have that
\begin{equation*}
|\hat{\mu}_v(r+1)h_v(\bar{t}_v)-\hat{\mu}_w(r+1)h_w(\bar{t}_w)|<
\frac{2\vartheta-1}{2}\,\|\bar{\rho}(r)\|+\vartheta\varepsilon\,.
\end{equation*}
We make a case distinction.
\begin{itemize}
  \item [Case 1:] $\mu_v(r+1)-\hat{\mu}_v(r+1)\leq \varepsilon$ and
  $\hat{\mu}_w(r+1)-\mu_w(r+1)\leq \varepsilon$. Because
  $\max\{h_v(\bar{t}_v),h_w(\bar{t}_w)\}\leq \vartheta$, we get
  \begin{align*}
  \mu_v(r+1)h_v(\bar{t}_v)-\mu_w(r+1)h_w(\bar{t}_w)
  &\leq (\mu_v(r+1)-\hat{\mu}_v(r+1))h_v(\bar{t}_v)\\
  &\qquad+\hat{\mu}_v(r+1)h_v(\bar{t}_v)-\hat{\mu}_w(r+1)h_w(\bar{t}_w)\\
  &\qquad+(\hat{\mu}_w(r+1)-\mu_w(r+1))h_w(\bar{t}_w)\\
  &\leq \frac{2\vartheta-1}{2}\,\|\bar{\rho}(r)\|+3\vartheta\varepsilon\,.
  \end{align*}
  \item [Case 2:] $\mu_v(r+1)-\hat{\mu}_v(r+1)>\varepsilon$. This implies that
  $\mu_v(r+1)=1\leq \mu_v(r)$.
  \begin{itemize}
    \item [a)] $\hat{\mu}_w(r+1)\leq \vartheta$, i.e., we have that
    $\mu_w(r+1)\geq \hat{\mu}_w(r+1)+\varepsilon$. Using
    \lemmaref{lem:freq_step}, we bound
    \begin{align*}
    \mu_v(r+1)h_v(\bar{t}_v)-\mu_w(r+1)h_w(\bar{t}_w)&\leq
    h_v(\bar{t}_v)\mu_v(r)-\left(\min_{u\in C}\{
    \mu_u(r)h_u(\bar{t}_u)\}+\frac{\varepsilon}{2}\right)\\
    &\leq \|\bar{\rho}(r)\|-\frac{\varepsilon}{2}\,.
    \end{align*}
    \item [b)] $\hat{\mu}_w(r+1)> \vartheta$, yielding that
    $\mu_w(r+1)\geq \vartheta-\varepsilon$. It follows that
    \begin{equation*}
    \mu_v(r+1)h_v(\bar{t}_v)-\mu_w(r+1)h_w(\bar{t}_w)
    \leq h_v(\bar{t}_v)-(\vartheta-\varepsilon) \leq \varepsilon\,.
    \end{equation*}
  \end{itemize}
  \item [Case 3:] $\hat{\mu}_w(r+1)-\mu_w(r+1)> \varepsilon$. This implies that
  $\mu_w(r+1)=\vartheta^2\geq \mu_w(r)$.
  \begin{itemize}
    \item [a)] $\hat{\mu}_v(r+1)> \vartheta$, i.e., we have that
    $\mu_v(r+1)\leq \hat{\mu}_v(r+1)-\varepsilon$. Using
    \lemmaref{lem:freq_step}, we bound
    \begin{align*}
    \mu_v(r+1)h_v(\bar{t}_v)-\mu_w(r+1)h_w(\bar{t}_w)&\leq
    \left(\max_{u\in C}\{
    \mu_u(r)h_u(\bar{t}_u)\}-\frac{\varepsilon}{2}\right)
    -h_w(\bar{t}_w)\mu_w(r)\\
    &\leq \|\bar{\rho}(r)\|-\frac{\varepsilon}{2}\,.
    \end{align*}
    \item [b)] $\hat{\mu}_v(r+1)\leq \vartheta$, yielding that
    $\mu_v(r+1)\leq \vartheta+\varepsilon$. It follows that
    \begin{equation*}
    \mu_v(r+1)h_v(\bar{t}_v)-\mu_w(r+1)h_w(\bar{t}_w)
    \leq (\vartheta+\varepsilon)h_v(\bar{t}_v)-\vartheta^2 \leq
    \vartheta\varepsilon\,.
    \end{equation*}
  \end{itemize}
\end{itemize}
In all cases, we get that
\begin{align*}
\max_{u,u'\in
C}\{|\mu_u(r+1)h_u(\bar{t}_u)-\mu_{u'}(r+1)h_{u'}(\bar{t}_{u'})\}
&=\mu_v(r+1)h_v(\bar{t}_v)-\mu_w(r+1)h_w(\bar{t}_w)\\
&\leq \max\left\{
\frac{2\vartheta-1}{2}\,\|\bar{\rho}(r)\|
+ 3\vartheta\varepsilon,
\|\bar{\rho}(r)\|-\frac{\varepsilon}{2}\right\}.\qedhere
\end{align*}
\end{proof}
It remains to take into account that hardware clock speeds change between rounds
using \lemmaref{lem:stability}.
\begin{corollary}\label{cor:freq_convergence}
For all $r\in \N$,
\begin{equation*}
\|\bar{\rho}(r+1)\|\leq \max\left\{
\frac{2\vartheta-1}{2}\,\|\bar{\rho}(r)\|
+ 3\vartheta\varepsilon, \|\bar{\rho}(r)\|-\frac{\varepsilon}{2}\right\}
+ 2\nu(T+\tau_2)\,.
\end{equation*}
\end{corollary}
\begin{proof}
By applying \lemmaref{lem:freq_convergence} and noting that for all $u\in C$,
$|\bar{\rho}(r)_v-\bar{\rho}(r+1)_v|\leq \nu (T+\tau_2)$ by
\lemmaref{lem:stability}.
\end{proof}
We conclude that the steady state frequency error is in $\BO(\varepsilon)$.
\begin{corollary}\label{cor:freq_error_bound}
Assume that $\beta:=(2\vartheta-1)/2<1$. Then
\begin{equation*}
\lim_{r\to \infty}\sup_{r'\geq r}\{\|\vec{\rho}(r')\|\}\leq
\frac{3\vartheta\varepsilon+2\nu(T+\tau_2)}{1-\beta}+\nu(T+\tau_2)\in
\BO(\varepsilon)\,.
\end{equation*}
\end{corollary}
\begin{proof}
From iterative application of \corollaryref{cor:freq_convergence}, we get that
\begin{equation*}
\lim_{r\to \infty}\sup_{r'\geq r}\{\|\vec{\rho}(r')\|\}\leq
\frac{3\vartheta\varepsilon+2\nu(T+\tau_2)}{1-\beta}\,.
\end{equation*}
\lemmaref{lem:stability} shows that $\|\vec{\rho}(r')\|\leq
\|\bar{\rho}(r')\|+\nu(T+\tau_2)$. Since \conditionref{cond:freq_constraints}
holds, $1-\beta\in \Omega(1)$ and the overall error is bounded by
$\BO(\varepsilon)$.
\end{proof}

\subsubsection*{Steady State Error with Frequency Correction}

To make use of \corollaryref{cor:freq_error_bound}, we need to derive a variant
of \corollaryref{cor:iteration} that allows for better control of
$\|\vec{p}(r+1)\|$ in case $\|\bar{\rho}(r)\|$ is small.

\begin{lemma}\label{lem:iteration_freq}
If round $r\in \N$ is executed correctly, then
\begin{equation*}
\|\vec{p}(r+1)\|\leq
\frac{4\bar{\vartheta}^2+5\bar{\vartheta}-7}{2(\bar{\vartheta}+1)}
\|\vec{p}(r)\|+\left(4\bar{\vartheta}-2\right)U
+\|\vec{\rho}(r)\|T\,.
\end{equation*}
\end{lemma}
\begin{proof}
For $v,w\in C$, assume w.l.o.g.\ that $p_v(r+1)-p_w(r+1)\geq 0$ (the other
case is symmetric). Denote by $\rho_v\in \vec{\rho}(r)_v$ the average
(adjusted) clock rate of $v$ during $[p_v(r),p_v(r+1)]$, i.e.,
\begin{equation*}
T-\Delta_v(r)=\frac{H_v(p_v(r+1))-H_v(p_v(r))}{\mu_v(r)}=\rho_v(p_v(r+1)-p_v(r))\,;
\end{equation*}
$\rho_w$ is defined analogously for $w$. Recall that $1\leq \rho_u\leq
\bar{\vartheta}$ for $u\in \{v,w\}$. Using this and \corollaryref{cor:step}
(with $\vartheta$ replaced by $\bar{\vartheta}=\vartheta^3$), we conclude that
\begin{align*}
&\quad~p_v(r+1)-p_w(r+1)\\
&= p_v(r)-p_w(r)+\frac{T-\Delta_v(r)}{\rho_v}
-\frac{T-\Delta_w(r)}{\rho_w}\\
&\leq p_v(r)-\Delta_v(r)-(p_w(r)-\Delta_w(r))
+\frac{\rho_w-\rho_v}{\rho_v\rho_w}\,T
+\left(1-\frac{1}{\rho_v}\right)|\Delta_v(r)|
+\left(1-\frac{1}{\rho_w}\right)|\Delta_w(r)|\\
&\leq \frac{5\bar{\vartheta}-3}{2(\bar{\vartheta}+1)}\|\vec{p}(r)\|+2\bar{\vartheta} U
+\|\vec{\rho}(r)\|T+2(\bar{\vartheta}-1)(\|\vec{p}(r)\|+U)\\
&=\frac{4\bar{\vartheta}^2+5\bar{\vartheta}-7}{2(\bar{\vartheta}+1)}
\|\vec{p}(r)\|+\left(4\bar{\vartheta}-2\right)U
+\|\vec{\rho}(r)\|T\,.\qedhere
\end{align*}
\end{proof}
Plugging this into our machinery we arrive at the main result of this section.
\begin{theorem}\label{thm:freq}
Suppose that \conditionref{cond:freq_constraints} is satisfied and that
\begin{equation*}
\bar{\alpha}:=\frac{2\bar{\vartheta}^2+5\bar{\vartheta}-5}{2(\bar{\vartheta}+1)}+(4\bar{\vartheta}+3)(\bar{\vartheta}-1)<1
\end{equation*}
(which is the case for $\vartheta\leq 1.01$). Then, with
$\alpha:=(4\bar{\vartheta}^2+5\bar{\vartheta}-7)/(2(\bar{\vartheta}+1))<1$ and
$\beta:=(2\vartheta-1)/2<1$, \algref{alg:frequency} has steady state error
\begin{equation*}
E\leq \frac{(4\bar{\vartheta}-2)U+\nu(T+\tau_2)T}{1-\alpha}
+\frac{(3\vartheta\varepsilon+2\nu(T+\tau_2))T}{(1-\alpha)(1-\beta)}\,.
\end{equation*}
\end{theorem}
\begin{proof}
As the preconditions of \theoremref{thm:basic_freq} are satisfied, all rounds
are executed correctly. By \corollaryref{cor:freq_error_bound}, this implies
that
\begin{equation*}
\lim_{r\to \infty}\sup_{r'\geq r}\{\|\vec{\rho}(r')\|\}\leq
\frac{3\vartheta\varepsilon+2\nu(T+\tau_2)}{1-\beta}+\nu(T+\tau_2)\,.
\end{equation*}
We plug this into the bound from \lemmaref{lem:iteration_freq}, which we apply
inductively to show that
\begin{align*}
E=\lim_{r\to \infty}\sup_{r'\geq r}\{\|\vec{p}(r')\|\}&\leq
\frac{(4\bar{\vartheta}-2)U+\lim_{r\to \infty}\sup_{r'\geq r}
\{\|\vec{\rho}(r)\|T\}}{1-\alpha}\\
&\leq \frac{(4\bar{\vartheta}-2)U+\nu(T+\tau_2)T}{1-\alpha}
+\frac{(3\vartheta\varepsilon+2\nu(T+\tau_2))T}{(1-\alpha)(1-\beta)}\,.\qedhere
\end{align*}
\end{proof}
Under reasonable assumptions we can obtain a more readable error bound.
\begin{corollary}\label{cor:freq}
Assume that the prerequisites of \theoremref{thm:freq} are satisfied. Moreover,
suppose that
\begin{itemize}
  \item $\alpha\approx 1/2$,
  \item $\varepsilon$ is chosen minimally such that it satisfies
  \conditionref{cond:freq_constraints},
  \item $T\approx \tau_3\gg \tau_2$, which is feasible whenever $T\gg
  \bar{\vartheta} (e(1)+d)$, and
  \item $\max\{(\bar{\vartheta}-1)^2T,\nu T^2\}\ll U$.
\end{itemize}
Then the steady state error of \algref{alg:frequency} is bounded by roughly
$28U$.
\end{corollary}
\begin{proof}
Note that $1/\alpha\approx 1/2$ implies that $1/\beta\approx 1/2$ and that
$\bar{\vartheta}\approx 1$. Plugging in $\varepsilon$ into the bound from
\theoremref{thm:freq}, the steady state error is approximately bounded by
\begin{align*}
4U+10\nu(T+\tau_2)T+12\varepsilon T
&\approx 
4U+10\nu(T+\tau_2)T
+12\left(6(\bar{\vartheta}-1)^2+\frac{2U}{\tau_2+\tau_3}+4\nu T\right) T\\
&\approx \left(4+\frac{24T}{\tau_2+\tau_3}\right)U+72(\bar{\vartheta}-1)^2T+58\nu
T^2\\
&\approx 28U\,.\qedhere
\end{align*}
\end{proof}
A few remarks:
\begin{itemize}
  \item \corollaryref{cor:freq} basically states that increasing $T$ is fine, as
  long as $\max\{(\bar{\vartheta}-1)^2T,\nu T^2\}\ll U$. This improves over
  \algref{alg:basic}, where it is required that $(\vartheta-1)T\ll U$, as it
  permits to transmit pulses at significantly smaller frequencies.
  \item While the error bound of roughly $28U$ is about factor $7$ larger than
  the about $4U$ \algref{alg:basic} provides, this is likely to be overly
  conservative. The source of this difference is that we assume that in a
  frequency measurement, the full uncertainty $U$ may skew the observation of
  the relative clock speed. However, this measurement is based on sending two
  signals in the same direction over the same communication link in fairly short
  order. In most settings, the difference in delays will be much smaller than
  between messages on \emph{different} communication links. Accordingly, the
  relative contribution of the frequency measurement to the error is likely to
  be much smaller in practice.
  \item If this is not the case, one may extend the time span for a
  frequency measurement over multiple rounds to decrease the effect of the
  uncertainty. This requires that the accumulated phase corrections do not
  become so large as to prevent a clear distinction of the frequency-related
  pulse (whose sending time must not be altered due to phase corrections) from
  phase-related pulses.\footnote{This issue can be circumvented by having a
  second, dedicated communication link between each pair of nodes.} To not
  further complicate the analysis, we refrained from presenting this option; it
  is used in~\cite{schossmaier98,schossmaier99}.
\end{itemize}

\section{Self-stabilization}\label{sec:self}

In this section, we propose a generic mechanism that can be used to transform
\algref{alg:basic} and \algref{alg:frequency} into \emph{self-stabilizing}
solutions. An algorithm is self-stabilizing, if it (re)establishes correct
operation from arbitrary states in bounded time. If there is an upper bound on
the time this takes in the worst case, we refer to it as the stabilization time.
We stress that, while self-stabilizing solutions to the problem are known, all
of them have skew $\Omega(d)$; augmenting the Lynch-Welch approach with
self-stabilization capabilities thus enables to achieve an optimal skew bound of
$\BO((\vartheta-1)T+U)$ in Byzantine self-stabilizing manner for the first time.

Our approach can be summarized as follows. Nodes locally count their pulses
modulo some $M\in \N$. We use a low-frequency, imprecise, but self-stabilizing
synchronization algorithm (called FATAL) from earlier
work~\cite{dolev14fatal,dolev14} to generate a ``heartbeat.'' On each such beat,
nodes will locally check whether the next pulse with number $1$ modulo $M$ will
occur within an expected time (local) window whose size is determined by the
precision the algorithm would exhibit after $M$ correctly executed pulses (in
the non-stabilizing case). If this is not the case, the node is ``reset'' such
that pulse $1$ will occur within this time window.

This simple strategy ensures that a beat forces all nodes to generate a pulse
with number $1$ modulo $M$ within a bounded time window. Assuming a value of $F$
corresponding to its length in \algref{alg:basic} or \algref{alg:frequency}
hence ensures that the respective algorithm will run as intended---at least up
to the point when the next beat occurs. Inconveniently, if the beat is not
synchronized with the next occurrence of a pulse $1 \bmod M$, some or all nodes
may be reset, breaking the guarantees established by the perpetual application
of approximate agreement steps. This issue is resolved by leveraging a feedback
mechanism provided by FATAL: FATAL offers a (configurable) time window during
which a NEXT signal externally provided to each node may trigger the next beat.
If this signal arrives at each correct node at roughly the same time, we can be
sure that the corresponding beat is generated shortly thereafter. This allows
for sufficient control on when the next beat occurs to prevent any node from
ever being reset after the first (correct) beat. Since FATAL stabilizes
regardless of how the externally provided signals behave, this suffices to
achieve stabilization of the resulting compound algorithm.

\subsection{FATAL}\label{sec:fatal}

\begin{algorithm}[p!]\label{alg:stab}
\caption{Interface algorithm, actions for node~$v\in C$ in response to a local
event at time $t$. Runs in parallel to local instances of FATAL and either
\algref{alg:basic} or \algref{alg:frequency}. In case \algref{alg:basic} is
used, we assume that $\tau_1(r)$, $\tau_2(r)$, and $T(r)$ do not depend on
$r\in \N$ and omit $r$ from the notation.}
\SetKwFunction{reset}{reset}
\SetKwBlock{func}{Function}{}
// algorithm maintains local variable $i\in \{0,\ldots,M-1\}$\\
\If{$v$ generates a pulse at time $t$}{
	$i:=i+1\bmod M$\;
	\If{$i=0$}{
		wait for local time $H_v(t)+\vartheta e(M)$\;
		trigger NEXT signal\;
	}
}
\If{$v$ generates a beat at time $t$}{
	\If{$i\neq 0$}{
		// beats should align with every $M^{th}$ pulse, hence reset\\
		\reset{$R^+$}\;
	}
	\ElseIf{next pulse would be sent before local time $H_v(t)+R^-$}{
		// reset to avoid early pulse\\
		\reset{$R^+-(H_v(t')-H_v(t))$}, where $t'$ is the current time\;
	}
	\ElseIf{next round has not started yet at local time $H_v(t)+R^+$}{
		// reset to avoid late pulse and start listening for other nodes' pulses on
		time\\
		\reset{$0$}\;
	}
}
\func(\reset{$\tau$}){
  halt local instance of clock synchronization algorithm\;
  wait for $\tau$ local time\;
  $i:=0$\;
  $H_v(t_v(0)):=H_v(t')$, where $t'$ is current time (i.e., $t_v(0):=t'$)\;
  restart loop of clock synchronization algorithm (in round $r=1$)\;
}
\end{algorithm}

\begin{figure}[p!]
	\centering
	\def\svgwidth{\columnwidth}
	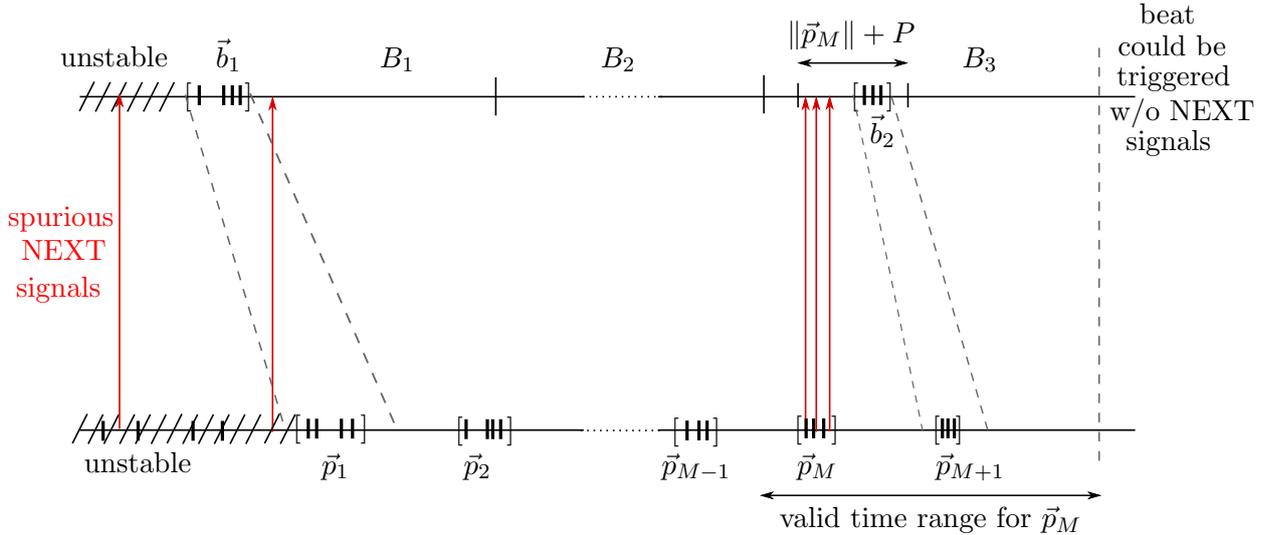
	\caption{Interaction of the beat generation and clock synchronization
	algorithms in the stabilization process, controlled by \algref{alg:stab}. Beat
	$\vec{b}_1$ forces pulse $\vec{p}_1$ to be roughly synchronized. The
	approximate agreement steps then result in tightly synchronized pulses. By the
	time the nodes trigger beat $\vec{b}_2$ by providing NEXT signals based on
	$\vec{p}_M$, synchronization is tight enough to guarantee that the beat
	results in no resets.}
   \label{fig:algorithm_stable}
\end{figure}

We summarize the properties of FATAL in the following corollary, where each node
has the ability to trigger a local NEXT signal perceived by the local instance
of FATAL at any time.
\begin{corollary}[of~\cite{dolev14fatal}]\label{cor:beats}
For suitable parameters $P,B_1,B_2,B_3,D\in \R^+$, FATAL stabilizes within
$\BO((B_1+B_2+B_3)n)$ time with probability $1-2^{-\Omega(n)}$. Once stabilized,
nodes $v\in C$ generate beats $b_v(k)$, $k\in \N$, such that the following
properties hold for all $k\in \N$.
\begin{enumerate}
  \item For all $v,w\in C$, we have that $|b_v(k)-b_w(k)|\leq P$.
  \item If no $v\in C$ triggers its NEXT signal during $[\min_{w\in
  C}\{b_w(k)\}+B_1,t]$ for some $t\leq \min_{w\in C}\{b_w(k)\}+B_1+B_2+B_3$,
  then $\min_{w\in C}\{b_w(k+1)\}\geq t$.
  \item If all $v\in C$ trigger their NEXT signals during $[\min_{w\in
  C}\{b_w(k)\}+B_1+B_2,t]$ for some $t\leq \min_{w\in C}\{b_w(k)\}+B_1+B_2+B_3$,
  then $\max_{w\in C}\{b_w(k+1)\}\leq t+P$.
\end{enumerate}
Denoting by $d_F$ the maximum end-to-end delay (sum of maximum message and
computational delay) of FATAL, for any $\phi\geq 1$ and any constant $C$ we
can ensure that
\begin{align*}
P&\in \BO(d_F)\\
B_1&\geq P+d\\
B_1+B_2+B_3&\in \Theta(\phi\cdot (d_F+d))\\
B_3&\geq C(B_1+B_2)\,.
\end{align*}
\end{corollary}
\begin{proof}
For $\phi=1$, all statements follow directly from Lemma~3.4 and Corollary~4.16
in~\cite{dolev14fatal}, noting that nodes will switch from state ready to
propose (in the main state machine) in response to a NEXT signal if their
timeout $T_3$ is expired. Once all correct nodes switched to propose, this
results in all nodes switching to accept and generating a beat within $d_F$
time. For $\phi>1$, one simply needs to observe that multiplying each timeout
for a satisfying Condition~3.3 in~\cite{dolev14fatal} by $\phi$ results in
another valid choice; the bound on the stabilization time given in
Corollary~4.16 scales accordingly.
\end{proof}

\subsection{Algorithm}\label{sec:stab_algo}

Our self-stabilizing solution utilizes both FATAL and the clock synchronization
algorithm with very limited interaction. We already stressed that FATAL will
stabilize regardless of the NEXT signals and note that it is not influenced by
\algref{alg:stab} in any other way. Concerning the clock synchronization
algorithm (either \algref{alg:basic} or \algref{alg:frequency}), we assume that
a ``careful'' implementation is used that does not maintain state variables for
a long time. Concretely, \algref{alg:basic} will clear memory between loop
iterations, and \algref{alg:frequency} will memorize the new multiplier value
$\mu_v(r+1)$ only, which is explicitly assigned during round $r$. If this is
satisfied, no further consistency checks of variables are required, and it will
be straightforward to re-use the analyses from Sections~\ref{sec:basic_analysis}
and~\ref{sec:freq_analysis}.

Having said this, let us turn to \algref{alg:stab}, which is basically an
ongoing consistency check based on the beats that resets the clock
synchronization algorithm if necessary. The feedback triggering the next beat in
a timely fashion is implemented by simply triggering the NEXT signal on each
$M^{th}$ beat, with a small delay ensuring that all nodes arrive in the same
round and have their counter variable $i$ reading $0$. The consistency checks
then ask for $i=0$ and the next pulse being triggered within a certain local
time window; if either does not apply, the reset function is called, ensuring
that both conditions are met.

\conditionref{cond:constraints_stab} lists the constraints on $R^-$ (the
minimum local time between a beat and local pulse $1\bmod M$), $R^+$ (the
respective maximum local time), and $M$ (the number of pulses between beats) --
the parameters of \algref{alg:stab} -- need to satisfy so that we can show that
the algorithm is guaranteed to stabilize.
\begin{condition}\label{cond:constraints_stab}
We require that
\begin{align}
P+R^++\tau_1-\frac{R^-}{\vartheta} &\leq e(1)\label{eq:initial_skew}\\
P + R^+ &\leq
\frac{R^-}{\vartheta}\label{eq:listen_on_time}\\
P + R^+ + \tau_1 + d &\leq
\frac{R^-+\tau_2}{\vartheta}\label{eq:receive_on_time}\\
P + d &\leq \frac{R^- - \tau_1}{\vartheta}\label{eq:no_early}\\
P+R^++T+\vartheta(e(1)+U) &\leq B_1+B_2\label{eq:beat_trivial}\\
P+\vartheta e(M)&\leq B_1\label{eq:B_1}\\
B_1+B_2&\leq
e(M)+(M-1)\left(\frac{T}{\vartheta}-\tau_1\right)+\frac{R^-}{\vartheta}
\label{eq:B_2}\\
\vartheta e(M)+(M-1)(T+\vartheta\tau_1)+P+R^++\tau_1
&\leq B_1+B_2+B_3\label{eq:B_3}\\
R^-  &\leq \frac{T}{\vartheta}-((\vartheta+2) e(M) + U + P)
\label{eq:no_early_round}\\
T+\vartheta(e(M)+U)-\tau_1 & \leq R^+\,.\label{eq:no_late_round}
\end{align}
\end{condition}
Intuitively, these constraints ensure the following:
\begin{itemize}
  \item \eqref{eq:initial_skew} says that resets on a beat enforce the skew to
  become bounded by $e(1)$.
  \item \eqref{eq:listen_on_time} and \eqref{eq:receive_on_time} ensure that
  correct nodes receive the first pulses from all other correct nodes after a
  beat.
  \item \eqref{eq:no_early} guarantees that these are actually the ``round-$1$''
  pulses also for nodes that have been reset, i.e., there are no spurious pulses
  from before such a reset that are received during the respective time window.
  \item \eqref{eq:beat_trivial} and \eqref{eq:B_1} make sure that FATAL will
  ignore any NEXT signals that may still be active when a beat occurs and that
  there is sufficient time for the first round after the beat to complete.
  \item \eqref{eq:B_2} and \eqref{eq:B_3} enforce that the (now correctly
  executing) algorithm will trigger the NEXT signals and thus the next beat
  well-aligned with the time reference it provides.
  \item Finally, \eqref{eq:no_early_round} and \eqref{eq:no_late_round} imply
  that such a beat will result in no resets.
\end{itemize}

We need to show that these constraints can be satisfied in conjunction with the
ones required by the employed synchronization algorithm.
\begin{lemma}\label{lem:solvable_stab}
Conditions~\ref{cond:constraints} and~\ref{cond:constraints_stab} can be
simultaneously satisfied such that $\tau_1(r)=\tau_1$, $\tau_2(r)=\tau_2$ and
$T(r)=T$ for all $r\in \N$, and $\lim_{r\to \infty} e(r)<\infty$ if
\begin{equation*}
\alpha=\frac{2\vartheta^2+\vartheta}{2-\vartheta}\cdot\left(1-\frac{1}{\vartheta^2}
+\frac{4(\vartheta-1)}{1-\beta}\right)<1\,,
\end{equation*}
where $\beta=(2\vartheta^2+5\vartheta-5)/(2(\vartheta+1))$. In this case,
\begin{equation*}
\lim_{r\to \infty} e(r) = \frac{(1-1/\vartheta)T+(3\vartheta
-1)U}{1-\beta}\,.
\end{equation*}
Here, we may choose any $T\geq T_0\in \BO((d_F+d)/(1-\alpha))$ and $B_1$,
$B_2$, and $B_3$ such that FATAL stabilizes in time $\BO(n (d_F+d))$ with
probability $1-2^{-\Omega(n)}$.
\end{lemma}
\begin{proof}
We choose $R^-$ and $R^+$ such that \eqref{eq:no_early_round} and
\eqref{eq:no_late_round} are satisfied with equality. Thus, any choice of
\begin{equation*}
F\geq \left(1-\frac{1}{\vartheta^2}\right)T+2P+4\vartheta e(M)+2\vartheta U
\end{equation*}
satisfies \eqref{eq:initial_skew}, and for
\eqref{eq:listen_on_time}--\eqref{eq:no_early} to hold it is sufficient that
\begin{align*}
F&\leq \tau_1 \leq \frac{T}{\vartheta}-3\vartheta
e(M)-\vartheta d-(\vartheta-1)P\\
\vartheta F &\leq \tau_2\,.
\end{align*}
These lower bounds on $\tau_1$ and $\tau_2$ are weaker than those imposed by
\conditionref{cond:constraints}, which demands that
$\min\{\tau_1,\tau_2\}\geq \vartheta e(1)>F$. Setting $\tau_1:=\vartheta e(1)$,
$\tau_2:=\vartheta(e(1)+d)$, and requiring
$T\geq\vartheta(\tau_1+\tau_2+e(1)+U)$ thus guarantees that the above lower
bounds on $\tau_1$ and $\tau_2$ hold, we have that
\begin{equation*}
\frac{T}{\vartheta}>\tau_1+F+\vartheta d>\tau_1+3\vartheta
e(M)+\vartheta d+(\vartheta-1)P\,,
\end{equation*}
and the inequalities of \conditionref{cond:constraints} are satisfied for $r=1$.
Moreover, with $x:=(3\vartheta-1)U+(1-1/\vartheta)T$, we have for $r\in \N$ that
\begin{equation*}
e(r)=\beta^{r-1}e(1)+\frac{1-\beta^{r-1}}{1-\beta}\,x\,,
\end{equation*}
i.e., $e(r)$ is a convex combination of $e(1)$ and $x/(1-\beta)$. We require
that $e(1)\geq x/(1-\beta)$, i.e.,
\begin{equation*}
\frac{F}{2-\vartheta}=e(1)\geq
\frac{(3\vartheta-1)U+(1-1/\vartheta)T}{1-\beta}\,;
\end{equation*}
here, we used that $2-\vartheta>0$, because $\alpha<1$. Thus, $e(r)\leq e(1)$,
and we conclude that \conditionref{cond:constraints} holds for
\begin{equation*}
F:=\max\left\{\left(1-\frac{1}{\vartheta^2}\right)T+2P+4\vartheta
e(M)+2\vartheta U,
\frac{(2-\vartheta)((3\vartheta-1)U+(1-1/\vartheta)T)}{1-\beta}\right\}
\end{equation*}
under the constraint that
\begin{equation*}
T\geq \vartheta(\tau_1+\tau_2+e(1)+U)=
\vartheta\left(\frac{(2\vartheta+1)F}{2-\vartheta}+\vartheta
d +U\right).
\end{equation*}
For any $c>1$, sufficiently large $M$ ensures that
\begin{equation*}
e(M)\leq c \lim_{r\to \infty}e(r) = \frac{cx}{1-\beta}=
\frac{c((3\vartheta-1)U+(1-1/\vartheta)T)}{1-\beta},
\end{equation*}
where the last step uses that $1-\beta\in \Omega(1)$ because $\alpha<1$.

Assuming sufficiently large $M$, the above lower bound on $T$ can hence be
met iff
\begin{equation*}
\frac{2\vartheta^2+\vartheta}{2-\vartheta}\cdot\max\left\{1-\frac{1}{\vartheta^2}
+\frac{4(\vartheta-1)}{1-\beta},
\frac{(2-\vartheta)(1-1/\vartheta)}{1-\beta}\right\}
=\alpha<1\,.
\end{equation*}
In this case, for sufficiently large $M$ the constraint on $T$ is satisfied if
\begin{equation*}
(1-\alpha)T \geq (1-\alpha)T_0\in \BO\left(\max\left\{P +
\frac{U}{1-\beta}+U,\frac{U}{1-\beta}\right\}+d+U\right)=\BO(P+d)\,,
\end{equation*}
where we used that $\vartheta$ and thus $1-\alpha$ and $1-\beta$ are constants.

To complete the proof, it remains to show that, for any such choice of $T$ and
a given lower bound on $M$, we can satisfy
Inequalities \eqref{eq:beat_trivial}--\eqref{eq:B_3} such that FATAL has the
claimed guarantees on the stabilization time. Given that all parameters
except for $M$, $B_1$, $B_2$, and $B_3$ are already fixed independently of these
values, it suffices if we can solve the system
\begin{align*}
K&\leq B_1\\
B_1+B_2&\leq (M-1)K\\
\vartheta M K&\leq B_1+B_2+B_3\\
\end{align*}
for an arbitrary $K\in \R^+$ such that $M$ is sufficiently large. By
\corollaryref{cor:beats}, we may choose $B_1$, $B_2$, and $B_3$ such that, e.g.,
$B_3\geq B_1+B_2$. Picking $\phi\geq 1$ in the corollary sufficiently large, we
get that $\phi B_1\geq K$ and $M:= \lfloor 2(B_1+B_2)/(\vartheta K)\rfloor$ is
sufficiently large and satisfies the second and third inequality (where again
we use that $2-\vartheta \in \Omega(1)$).

Finally, note that $P\in \BO(d_F)$ and all factors occurring in this proof are
constants depending on $\vartheta$ only, implying that $\phi$ and $M$ are
constants as well. The bound on the stabilization time thus readily follows from
\corollaryref{cor:beats} as well.
\end{proof}

In the remainder of the section, we assume (i) that the beat generation
algorithm has already stabilized, i.e., the guarantees stated in
\corollaryref{cor:beats} hold, (ii) that the executed clock synchronization
algorithm is \algref{alg:basic}, and (iii) that \conditionref{cond:constraints}
holds. The analysis for \algref{alg:frequency} is analogous, where
$\bar{\vartheta}=\vartheta^3$ takes the role of $\vartheta$ and
\conditionref{cond:freq_constraints} takes the role of
\conditionref{cond:constraints}; this is formalized by the following corollary
and \theoremref{thm:stab_freq} at the end of this section.

\begin{corollary}\label{cor:solvable_stab}
Conditions~\ref{cond:freq_constraints} and~\ref{cond:constraints_stab} can be
simultaneously satisfied such that $\lim_{r\to \infty} e(r)<\infty$ if
\begin{equation*}
\bar{\alpha}=\frac{4\bar{\vartheta}^2+5\bar{\vartheta}}{2-\bar{\vartheta}}\cdot\left(1-\frac{1}{\bar{\vartheta}^2}
+\frac{4(\bar{\vartheta}-1)}{1-\bar{\beta}}\right)<1\,,
\end{equation*}
where $\bar{\vartheta}=\vartheta^3$ and
$\bar{\beta}=(2\bar{\vartheta}^2+5\bar{\vartheta}-5)/(2(\bar{\vartheta}+1))$. In this case,
\begin{equation*}
\lim_{r\to \infty} e(r) = \frac{(1-1/\bar{\vartheta})T+(3\bar{\vartheta}
-1)U}{1-\beta}\,.
\end{equation*}
Here, we may choose any $T\geq T_0\in \BO((d_F+d+U)/(1-\alpha))$ and $B_1$,
$B_2$, and $B_3$ such that FATAL stabilizes in time $\BO(n (d_F+d))$ with
probability $1-2^{-\Omega(n)}$.
\end{corollary}
\begin{proof}
Analogous to the proof of \lemmaref{lem:solvable_stab}, but replacing the
constraint $T\geq \vartheta(\tau_1+\tau_2+e(1)+U)$ by $T\geq
\tau_1+\tau_2+\tau_3+\tau_4+\bar{\vartheta}(e(1)+U)>\bar{\vartheta}(\tau_1+\tau_2+e(1)+U)$
and setting $\tau_3:=\bar{\vartheta}(e(1)+(1-1/\bar{\vartheta})(\tau_1+\tau_2))$ and
$\tau_4:=\bar{\vartheta}(e(1)+d+(1-1/\bar{\vartheta})(\tau_1+\tau_2))$ in accordance with
\conditionref{cond:freq_constraints}. This results in the requirement that
\begin{equation*}
T\geq \frac{(4\bar{\vartheta}^2+5\bar{\vartheta})F}{2-\vartheta}+\bar{\vartheta}
d +U\,,
\end{equation*}
which in turn leads to the value for $\bar{\alpha}$.
\end{proof}

\subsection{Analysis}\label{sec:stab_analysis}

Our analysis starts with the first correct beat produced by FATAL, which is
perceived at node $v\in C$ at time $b_v(1)$. Subsequent beats at $v$ occur at
times $b_v(2)$, $b_v(3)$, etc. We first establish that the first beat guarantees
to ``initialize'' the synchronization algorithm such that it will run correctly
from this point on (neglecting for the moment the possible intervention by
further beats). We use this do define the ``first'' pulse times $p_v(1)$, $v\in
C$, as well; we enumerate consecutive pulses accordingly.
\begin{lemma}\label{lem:init_stab}
Let $b:=\min_{v\in C}\{b_v(1)\}$. We have that
\begin{enumerate}
\item Each $v\in C$ generates a pulse at time $p_v(1)\in
[b+R^-/\vartheta,b+P+R^++\tau_1]$.
\item $\|\vec{p}(1)\|\leq e(1)$.
\item At time $p_v(1)$, $v\in C$ sets $i:=1$.
\item $w\in C$ receives the pulse sent by $v\in C$ at a local
time from $[H_w(p_w(1))-\tau_1,H_w(p_w(1))+\tau_2]$.
\item This is the only pulse $w$ receives from $v$ at a
local time from $[H_w(p_w(1))-\tau_1,H_w(p_w(1))+\tau_2]$.
\item Denoting by round~$1$ the execution of the for-loop in \algref{alg:basic}
during which each $v\in C$ sends the pulse at time $p_v(1)$, this round is
executed correctly.
\end{enumerate}
\end{lemma}
\begin{proof}
Assume for the moment that $\min_{v\in C}\{b_v(2)\}$ is sufficiently large,
i.e., no second beat will occur at any correct node for the times relevant to
the proof of the lemma; we will verify this at the end of the proof.

From the pseudocode given in Algorithms~\ref{alg:basic} and~\ref{alg:stab}, it
is straightforward to verify that $v\in C$ generates a pulse at a local time
from $[H_v(b_v(1))+R^-,H_v(b_v(1))+R^++\tau_1]$. Since $b_v(1)\in [b,b+P]$ by
\corollaryref{cor:beats}, this shows the first claim. The second follows
immediately, since
\begin{equation*}
\|\vec{p}(1)\|\leq P+R^++\tau_1-\frac{R^-}{\vartheta}
\stackrel{\eqref{eq:initial_skew}}{\leq} e(1)\,.
\end{equation*}

Note that, until we show the last claim, it is not clear that $p_v(1)$ is unique
for each $v\in C$. For the moment, let $p_v(1)$ be the first pulse $v\in C$
sends during the local time interval $[H_v(b_v(1))+R^-,H_v(b_v(1))+R^++\tau_1]$.
With this convention, the third claim is shown as follows. Observe that any
$v\in C$ that executes the reset function in response to the beat sets $i:=0$
when doing so. Hence, it will set $i:=1$ at time $p_v(1)$. Thus, consider $v\in
C$ that does not execute the reset function. This entails that $i=0$ at time
$b_v(1)$ and $v$ generates no pulse during local times from
$[H_v(b_v(1),H_v(b_v(1))+R^-)$. Consequently, $v$ will increase $i$ to $1$ at
time $p_v(1)$.

For the fourth claim, we bound
\begin{equation*}
p_v(1)\geq b+\frac{R^-}{\vartheta}\geq
b_w(1)+\frac{R^-}{\vartheta}-P
\stackrel{\eqref{eq:listen_on_time}}{\geq}b_w(1)+R^+\,.
\end{equation*}
Thus, either the next round has already started at node $w$ by time $p_v(1)$ or
$w$ calls reset with argument $0$, i.e., starts a new round. Either way, we have
that $w$ receives the pulse from $v$ no earlier than local time
$H_w(p_w(1))-\tau_1$. To see that the pulse arrives on time, we bound
\begin{equation*}
p_v(1)+d\leq p_w(1)+P+R^++\tau_1+d-\frac{R^-}{\vartheta}
\stackrel{\eqref{eq:receive_on_time}}{\leq}p_w(1)+\frac{\tau_2}{\vartheta}\,.
\end{equation*}
As $H_w(p_w(1)+\tau_2/\vartheta)\leq H_w(p_w(1))+\tau_2$, the fourth claim
follows.

Concerning the fifth claim, observe that $v\in C$ sends exactly one pulse during
the local time interval $[H_v(b_v(1)),H_v(p_v(1))]$. As for $w\in C$ we have
that
\begin{equation*}
b_v(1)+d\leq b_w(1)+P+d\leq p_w(1)-\frac{R^-}{\vartheta}+P+d
\stackrel{\eqref{eq:no_early}}{\leq} p_w(1)-\frac{\tau_1}{\vartheta}\,,
\end{equation*}
no pulse $v$ sent at an earlier local time is received by $w$ at or after local
time $H_w(p_w(1))-\tau_1$. In particular, the first pulse $w$ receives from $v$
at a local time from $[H_w(p_w(1))-\tau_1,H_w(p_w(1))+\tau_2]$ arrives at $w$ at
a time $t_{vw}\in [p_v(1)+d-U,p_v(1)+d]$. Since we also showed that
$\|\vec{p}(1)\|\leq e(1)$, we conclude that the analysis of
\sectionref{sec:basic_analysis} can be applied to show that any subsequent
pulse arrives after the round is complete at all nodes. Furthermore, we
conclude that round $1$ is executed correctly.

Recall that in the above reasoning, we assumed that $\min_{v\in C}\{b_v(2)\}$ is
sufficiently large. Clearly, this is the case if round $1$ ends at all nodes
before this time. Accordingly, we bound for $v\in C$
\begin{equation*}
p_v(1)+T-\Delta_v(1)-\tau_1
\leq b_v(1)+R^++T-\Delta_v(1)
\leq b+P+R^++T+\vartheta(e(1)+U)
\stackrel{\eqref{eq:beat_trivial}}{\leq}b+B_1+B_2\,,
\end{equation*}
where the second last step makes use of \corollaryref{cor:step}. Because no node
$v\in C$ generates a pulse with $i=M$ during times $[b_v(1)+\vartheta
e(M),p_v(2)]$, no such node triggers a NEXT signal during this time interval
(cf.~\algref{alg:stab}). We have that
\begin{equation*}
b_v(1)+\vartheta e(M)\leq b+P+\vartheta e(M)\stackrel{\eqref{eq:B_1}}{\leq}
B_1\,,
\end{equation*}
implying by \corollaryref{cor:beats} that $\min_{v\in C}\{b_v(2)\}\geq
b+B_1+B_2$.
\end{proof}

\lemmaref{lem:init_stab} serves as induction anchor for the argument showing
that all rounds of the algorithm are executed correctly. However, due to
possible interference of future beats, for the moment we can merely conclude
that this is the case until the next beat; we obtain the following corollary.
\begin{corollary}\label{cor:init_stab}
Denote by $N$ the infimum over all times $t\geq b+B_1$ at which some $v\in C$
triggers a NEXT signal. If $\min_{v\in C}\{p_v(M)+e(M)\}\leq
\min\{N,b+B_1+B_2+B_3\}$, then all rounds $r\in \{1,\ldots,M\}$ are executed
correctly and $\|\vec{p}(r)\|\leq e(r)$.
\end{corollary}
\begin{proof} \lemmaref{lem:init_stab} shows that the first beat ``initializes''
the system such that $\|\vec{p}(1)\|\leq e(1)$ and the first round is executed
correctly. By \corollaryref{cor:beats}, $\min_{v\in C}\{b_v(2)\}\geq
\min\{N,b+B_1+B_2+B_3\}$. Hence, after round $1$ \algref{alg:basic} will be
executed without interference from \algref{alg:stab} until (at least) time
$\min_{v\in C}\{p_v(M)+e(M)\}$. For $r\in \{2,\ldots,M\}$, the claim thus
follows as in \sectionref{sec:basic_analysis}.
\end{proof}

Next, we leverage this insight to prove that the progress of the synchronization
algorithm -- which will operate correctly at least until the next beat --
together with the constraints of \conditionref{cond:constraints_stab} ensures
the following: the first time when node $v\in C$ triggers its NEXT signal after
time $b+B_1$ falls within the window of opportunity for triggering the next beat
provided by FATAL.
\begin{lemma}\label{lem:N}
For $v\in C$, denote by $N_v(1)$ the infimum of times $t\geq b+B_1$ when it
triggers its NEXT signal. We have that $H_v(N_v(1))=p_v(M)+\vartheta e(M)$ and
that
\begin{equation*}
b+B_1+B_2\leq N_v(1)\leq b+B_1+B_2+B_3\,.
\end{equation*}
\end{lemma}
\begin{proof}
At time $b_v(1)$, $v\in C$ sets $i:=0$ (unless it already holds that $i=0$).
Thus, $v$ will not trigger the NEXT signal until it sent at least $M$ pulses
and waited for $\vartheta e(M)$ local time, i.e., $N_v(1)\geq p_v(M)+e(M)$.
As observed in the proof of \lemmaref{lem:init_stab}, we have that $b_v(1)\geq
b+B_1$. Thus, we can apply \corollaryref{cor:init_stab}, where
\begin{equation*}
N:=\min_{v\in C}\{N_v(1)\}\geq \min_{v\in C}\{p_v(M)+e(M)\}\,,
\end{equation*}
to conclude that one of the following must hold true: (i) all rounds $r\in
\{1,\ldots,M\}$ are executed correctly or (ii) $\min_{v\in
C}\{p_v(M)+e(M)\}>b+B_1+B_2+B_3$. 

In the first case, we have that 
\begin{equation*}
H_v(N_v(1))=H_v(p_v(1))+\vartheta e(M)+\sum_{r=1}^{M-1} T-\Delta_v(r)\,,
\end{equation*}
where 
\begin{equation*}
\sum_{r=1}^{M-1}|\Delta_v(r)|\leq \sum_{r=1}^{M_1}e(r)\leq
\vartheta(M-1)\tau_1\,.
\end{equation*}
We conclude that
\begin{equation*}
p_v(1)+e(M)+(M-1)\left(\frac{T}{\vartheta}-\tau_1\right)
\leq N_v(1)\leq p_v(1) +\vartheta e(M)+(M-1)(T+\vartheta\tau_1).
\end{equation*}
Applying the first statement of \lemmaref{lem:init_stab}, this yields that
\begin{equation*}
b+e(M)+(M-1)\left(\frac{T}{\vartheta}-\tau_1\right)+\frac{R^-}{\vartheta}
\leq N_v(1)\leq b +\vartheta e(M)+(M-1)(T+\vartheta\tau_1)+P+R^++\tau_1\,.
\end{equation*}
The claim now follows from \eqref{eq:B_2} and \eqref{eq:B_3}.

With respect to the second case, observe that since no NEXT signal is triggered
at any $v\in C$ after time $b+B_1$ until time $b+B_1+B_2+B_3$, $\min_{v\in
C}\{b_v(2)\}\geq b+B_1+B_2+B_3$ by \corollaryref{cor:beats}. Thus,
\algref{alg:basic} runs without interference up to this time. Using this, we can
establish the same bounds as for the first case.
\end{proof}
This immediately implies that the second beat occurs in response to the NEXT
signals, which itself are aligned with pulse $M$.
\begin{corollary}\label{cor:N}
For all $v\in C$, $b_v(2)\in [p_v(M),p_v(M)+(\vartheta+1)e(M)+P]$.
\end{corollary}
\begin{proof}
By \lemmaref{lem:N}, $N_v(1)\in [b+B_1+B_2,b+B_1+B_2+B_3]$ for all $v\in C$.
Thus, by \corollaryref{cor:init_stab}, $\|\vec{p}(M)\|\leq e(M)$. As $v\in C$
triggers its NEXT signal at local time $H_v(p_v(M))+\vartheta e(M)$, it follows
that
\begin{equation*}
p_v(M)\leq \min_{w\in C}\{p_w(M)+e(M)\}\leq \min_{w\in C}\{N_w(1)\}
\end{equation*}
and that
\begin{equation*}
\max_{w\in C}\{N_w(1)\}\leq \max_{w\in C}\{p_w(M)+\vartheta e(M)\}
\leq p_v(M)+(\vartheta+1)e(M)\,.
\end{equation*}
The claim now follows from the second and third statements of
\corollaryref{cor:beats}.
\end{proof}
Having established this timing relation between $\vec{b}(2)$ and $\vec{p}(M)$,
we can conclude that no correct node is reset due to the second beat.
\begin{lemma}\label{lem:no_reset}
Node $v\in C$ does not call the reset function of \algref{alg:stab} in response
to beat $b_v(2)$.
\end{lemma}
\begin{proof}
By \corollaryref{cor:N}, $b_v(2)\in [p_v(M),p_v(M)+(\vartheta+1)e(M)+P]$. By
\corollaryref{cor:init_stab}, \algref{alg:basic} has been executed without
interruption by beat after time $b_v(1)$ up to this time. Hence, $v$ sets $i:=M
\bmod M = 0$ at time $p_v(M)\leq b_v(2)$. As also round $M$ is executed
correctly, the earliest time when $v$ could generate pulse $M+1$ without a reset
is bounded by
\begin{align*}
p_v(M)+\frac{T-\Delta_v(M)}{\vartheta}&\geq
p_v(M)-(e(M)+U)+\frac{T}{\vartheta}\\
&\geq b_v(2)-((\vartheta+2)e(M)+P+U)+\frac{T}{\vartheta}\\
&\stackrel{\eqref{eq:no_early_round}}{\geq} b_v(2)+R^-\,,
\end{align*}
where in the first step we applied \corollaryref{cor:step}. This implies that
node $v$'s variable $i$ equals $0$ at time $b_v(2)$ and $v$ does not generate a
pulse at a local time from $[H_v(b_v(2)),H_v(b_v(2))+R^-]$. It remains to show
that $v$ enters round $M+1$ at the latest at local time $H_v(b_v(2))+R^+$. To
show this, we bound
\begin{align*}
H_v(p_v(M))+T-\tau_1-\Delta_v(M)&\leq
H_v(p_v(M))+T-\tau_1+\vartheta(e(M)+U)\\
&\leq H_v(b_v(2))+T-\tau_1+\vartheta(e(M)+U)\\
&\stackrel{\eqref{eq:no_late_round}}{\leq} b_v(2)+R^+\,.\qedhere
\end{align*}
\end{proof}
Repeating the above reasoning for all pairs of beats $\vec{b}(k)$,
$\vec{b}(k+1)$, $k\in \N$, it follows that no correct node is reset by any beat
other than the first. Thus, the clock synchronization algorithm is indeed
(re-)initialized by the first beat to run without any further meddling from
\algref{alg:stab}. This implies the same bounds on the steady state error as for
the original synchronization algorithm.
\begin{theorem}\label{thm:stab}
Suppose that \algref{alg:approxAgree} is executed with \algref{alg:basic} as
synchronization algorithm. If
\begin{equation*}
\alpha=\frac{2\vartheta^2+\vartheta}{2-\vartheta}\cdot\left(1-\frac{1}{\vartheta^2}
+\frac{4(\vartheta-1)}{1-\beta}\right)<1
\end{equation*}
(which holds for $\vartheta\leq 1.03$), where
$\beta=(2\vartheta^2+5\vartheta-5)/(2(\vartheta+1))$, then all parameters can be
chosen such that the compound algorithm is self-stabilizing and has steady state
error
\begin{equation*}
E \leq \frac{(\vartheta-1)T+(3\vartheta-1)U}{1-\beta}\,.
\end{equation*}
Here, any nominal round length $T\geq T_0\in \BO(d_F+d)$ is possible.
\end{theorem}
\begin{proof}
\lemmaref{lem:solvable_stab} that Conditions~\ref{cond:constraints}
and~\ref{cond:constraints_stab} can be satisfied such that $\lim_{r\to \infty}
e(r)=((\vartheta-1)T+(3\vartheta-1)U)/\beta$ and $T_0\in \BO(d_F+d)$. Hence, we
may apply the statements derived in this section.

By \corollaryref{cor:beats}, the beat generation mechanism will eventually
stabilize. Afterwards, we can apply \lemmaref{lem:no_reset} to show that the
second (correct) beat results in no calls to the reset function in
\algref{alg:stab}. In fact, this extends to any beat except for the first:
letting beat $k\in \N$ take the role of beat $1$, our reasoning shows that beat
$k+1$ does not result in a reset at any node. Moreover, applying the same
reasoning to \corollaryref{cor:init_stab}, we conclude that all rounds $r\in \N$
are executed correctly, and that $\|\vec{p}(r)\|\leq e(r)$. The bound on $E$
follows.
\end{proof}
Observe that, in comparison to \theoremref{thm:basic}, the expression obtained
for the steady state error replaces $d$ by $\BO(d_F+d)$, which is essentially
the skew upon initialization by the first beat. In \algref{alg:basic}, we
circumvented any dependence on $F$ by varying round lengths over time. For the
self-stabilizing solution, this is not possible, since counting rounds locally
is not guaranteed to ensure a consistent opinion across all nodes concerning the
nominal length of the current round; we are restricted to counting rounds
$\bmod M\in \N$, so any long round length will reoccur regularly.

It remains to draw the analogous conclusions for using \algref{alg:stab} with
\algref{alg:frequency} as synchronization algorithm.
\begin{theorem}\label{thm:stab_freq}
Suppose that \algref{alg:approxAgree} is executed with \algref{alg:frequency} as
synchronization algorithm. If
\begin{equation*}
\bar{\alpha}=\frac{4\bar{\vartheta}^2+5\bar{\vartheta}}{2-\bar{\vartheta}}\cdot\left(1-\frac{1}{\bar{\vartheta}^2}
+\frac{4(\bar{\vartheta}-1)}{1-\bar{\beta}}\right)<1
\end{equation*}
(which holds for $\vartheta\leq 1.004$), where $\bar{\vartheta}=\vartheta^3$ and
$\bar{\beta}=(2\bar{\vartheta}^2+5\bar{\vartheta}-5)/(2(\bar{\vartheta}+1))$,
then all parameters can be chosen such that the compound algorithm
self-stabilizes in $\BO(n)$ time and has steady state error
\begin{equation*}
E\leq \frac{(4\bar{\vartheta}-2)U+\nu(T+\tau_2)T}{1-\alpha}
+\frac{(3\vartheta\varepsilon+2\nu(T+\tau_2))T}{(1-\alpha)(1-\beta)}\,,
\end{equation*}
where $\alpha:=(4\bar{\vartheta}^2+5\bar{\vartheta}-7)/(2(\bar{\vartheta}+1))<1$
and $\beta:=(2\vartheta-1)/2<1$. Here, any value of $T\geq T_0\in \BO(d_F+d)$ is
possible.
\end{theorem}
\begin{proof}
As for \theoremref{thm:stab}, with \corollaryref{cor:solvable_stab} taking the
place of \lemmaref{lem:solvable_stab} and noting that the convergence argument
for the frequencies relies on rounds being executed correctly only (i.e., no
assumptions on $\mu_v(1)$, $v\in C$, are required).
\end{proof}
We remark that despite the stringent requirements on $\vartheta$ for the
recovery argument to work (i.e., $\bar{\alpha}<1$), the actual bound on the
precision involves $\alpha$ and $\beta$. If $\vartheta\leq 1.004$, we have
$\alpha\leq 0.512$ and $\beta\leq 0.502$. Concerning stabilization, we remark
that it takes $\BO(n)$ time with probability $1-2^{-\Omega(n)}$, which is
directly inherited from FATAL. The subsequent convergence to small skews is not
affected by $n$, and will be much faster for realistic parameters, so we refrain
from a more detailed statement.

\section{Conclusions}\label{sec:conclusions}

The results derived in this paper demonstrate that the Lynch-Welch
synchronization principle is a promising candidate for reliable clock
generation, not only in software, but also in hardware. Apart from accurate
bounds on the synchronization error depending on the quality of clocks, we
present a generic coupling scheme enabling to add self-stabilization properties.

We believe these results to be of practical merit. Concretely, first results
from a prototype Field-Programmable Gate Array (FPGA) implementation of
\algref{alg:basic} show a skew of $182\,$ps~\cite{HKL16}. Given the appealing
simplicity of the presented algorithms and this excellent performance, we
consider the approach a viable candidate for reliable clock generation in
fault-tolerant low-level hardware and other areas.

\section*{Acknowledgements}

We thank Matthias F\"ugger and Attila Kinali for fruitful discussions, and the
anonymous reviewers of an earlier version for valuable comments.

\bibliographystyle{plainnat}
\bibliography{sync}

\end{document}